\documentclass[a4paper,12pt]{article}

\usepackage{titling}
\usepackage{fontenc,fontaxes}
\usepackage{graphicx,booktabs}
\usepackage{mathtools}
\usepackage{amsmath,amsthm,amsfonts}
\usepackage{xspace}
\usepackage{authblk}
\usepackage{xurl}

\usepackage{comment}
\usepackage{setspace}
\usepackage[bookmarksnumbered,bookmarksopen,colorlinks,citecolor={blue}]{hyperref}
\usepackage{bbm}
\usepackage[bf,small]{caption}
\usepackage{doi}
\usepackage{subcaption}
\usepackage{tikz,pgfplots}
\pgfplotsset{compat=1.7}
\pgfplotsset{compat=newest}
\usepackage[margin=1.25in]{geometry}
\usepackage[inline,shortlabels]{enumitem}
\setlist[enumerate,1]{label={\upshape(\roman*)}}

\usepackage[authoryear,round]{natbib}

\mathtoolsset{mathic=true}

\theoremstyle{plain}

\newtheorem*{thm*}{Theorem}
\newtheorem{lem}{Lemma}
\newtheorem*{lem*}{Lemma}
\newtheorem{prop}{Proposition}
\newtheorem*{prop*}{Proposition}

\newtheorem*{asmp*}{Assumption}

\theoremstyle{definition}

\newtheorem*{defn*}{Definition}

\numberwithin{equation}{section}
\numberwithin{thm}{section}
\numberwithin{prop}{section}
\numberwithin{lem}{section}

\newcommand{\cN}{\mathcal{N}}

\newcommand{\cV}{\mathcal{V}}

\newcommand{\A}{\mathrm{A}}
\newcommand{\E}{\mathrm{E}}

\newcommand{\R}{\mathbb{R}}
\newcommand{\N}{\mathbb{N}}
\newcommand{\ubar}[1]{\text{\b{$#1$}}}
\newcommand\argmax{\operatornamewithlimits{arg\,max}}

\renewcommand\footnotemark{}

\begin{document}
	
	\title{Optimal taxation and the Domar-Musgrave effect}
	
	\author[1]{Brendan K.\ Beare}
	\author[2]{Alexis Akira Toda}\thanks{We thank John Coleman, as well as seminar participants at Duke University, Georgetown University, Pompeu Fabra University, the University of Connecticut, the University of Oxford and the University of Sydney, for helpful comments.}
	\affil[1]{School of Economics, University of Sydney}
	\affil[2]{Department of Economics, Emory University}
	
	\maketitle
	
	\begin{center}
		Accepted for publication in \textit{Economic Inquiry}.
	\end{center}
	
	\medskip
	
	\begin{abstract}
		This article concerns the optimal choice of flat taxes on labor and capital income, and on consumption, in a tractable economic model in which agents are subject to idiosyncratic investment risk. We identify the tax rates which maximize welfare in stationary equilibrium while preserving tax revenue, finding that an increase in welfare equivalent to a permanent increase in consumption of nearly 7\% can be achieved by only taxing capital income and consumption. The Domar-Musgrave effect explains cases where it is optimal to tax capital income. We characterize the dynamic response to the substitution of consumption taxation for labor income taxation.
	\end{abstract}
	
	\onehalfspacing
	
	\section{Introduction}\label{sec:intro}
	
	This article concerns the optimal choice of tax rates in an economy in which the primary source of heterogeneity is idiosyncratic variation in the return to capital. We develop a tractable model in which there are three forms of taxation: labor and capital income taxes, and a consumption tax. Each of the three taxes is applied at a flat rate which does not vary over time. The labor and capital income tax rates are between zero and one, and the consumption tax rate is nonnegative. Our model is a heterogeneous agent economy of the general type introduced in \cite{Bewley1986}, with idiosyncratic variation in the return to capital obtained by endowing agents with a Markov-switching entrepreneurial ability state. Agents supply labor inelastically, and may choose to engage in entrepreneurship by hiring labor to operate physical capital. They manage a portfolio of risk-free bonds and physical capital in the face of uncertainty about their lifespan and future returns to investment. The risk-free interest rate and wage are determined by market clearing conditions for the bond and labor markets. We characterize the combination of labor income, capital income and consumption tax rates which maximizes welfare in stationary equilibrium while generating a fixed level of revenue, finding that it is optimal to generate all revenue through the taxation of capital income and consumption. The optimal rate of capital income taxation is zero if entrepreneurs will fully leverage their physical capital in the absence of capital income taxation, or positive otherwise.
	
	A crucial feature of our model is that capital income taxation is applied with full offset provisions. This means that, just as a capital gain of \$10 is reduced to \$9 with a 10\% capital income tax, a capital loss of \$10 is also reduced to \$9. With full offset provisions and a positive expected return to capital investment, capital income taxation lowers the expected return but also reduces downside risk. The latter effect mitigates, and may even dominate, the extent to which the former effect discourages capital accumulation. The potential for capital income taxation to encourage capital investment in this fashion is known as the Domar-Musgrave effect, after \cite{DomarMusgrave1944}. Subsequent to this work, \cite{Tobin1958} provided a striking numerical example in which the introduction of a 50\% capital gains tax with full offset provisions leads an investor to double their investment in a risky asset. \cite{Mossin1968} and \cite{Stiglitz1969} extended the analysis in \cite{DomarMusgrave1944} from mean-variance preferences to general expected utility preferences. Recent articles showing the relevance of the Domar-Musgrave effect in real data include \cite{LangenmayrLester2018} and \cite{ArmstrongGlaeserHuangTaylor2019}.
	
	Our model closely resembles the one used in \cite{Panousi2010} to study the optimal choice of labor and capital income tax rates, itself a variation on the model introduced in \cite{Angeletos2007}. We provide a detailed analysis of our model going beyond what is provided in those articles. Specifically, we provide explicit formulae for the Mellin transform\footnote{The Mellin transform of the distribution of a positive random variable $X$ is $\E(X^z)$, viewed as a complex-valued function of a complex variable $z$.} of the stationary distribution of wealth, for the excess demand equations determining equilibrium prices, and for equilibrium aggregate tax revenue. We do this by applying new results in \cite{BeareToda2022} on a class of stochastic processes called \emph{Markov multiplicative processes with reset}. As in \cite{Panousi2010}, we model the creation and destruction of agents using the perpetual youth framework of \cite{Yaari1965} and \cite{Blanchard1985} in which agents perish with a fixed probability each period, at this time being replaced with a newborn agent. The random replacement of agents, combined with the multiplicative nature of investment returns, leads to the time path of the wealth of a succession of agents being a Markov multiplicative process with reset. We show, using the results in \cite{BeareToda2022}, that the wealth process admits a unique stationary distribution whose Mellin transform is given in closed form by model parameters. From this Mellin transform we obtain formulae for excess demands and equilibrium tax revenue, and can compute the stationary distribution of wealth by Fourier inversion.
	
	The formulae we derive facilitate a precise numerical analysis of our model illuminating previously unknown facets of the role played by the Domar-Musgrave effect. In our numerical analysis we assume the productivity of entrepreneurs to be distributed independently over time, while maintaining Markov switching between worker and entrepreneur types. The serial independence of entrepreneurial productivity leads all entrepreneurs to be affected by the natural borrowing constraint in the same way: either it binds on all of them (which we understand to mean that physical capital is fully leveraged) or it binds on none of them. The way in which the optimal tax rates relate to model parameters is qualitatively affected by whether the natural borrowing constraint for entrepreneurs is strictly binding, barely binding, or slack. By \emph{barely binding}, we mean that the amount that entrepreneurs would borrow in the absence of a borrowing constraint is exactly equal to the limit provided by the natural borrowing constraint. This is no mere knife-edge case, and obtains over a positive measure set of parameter values when tax rates are chosen optimally.
	
	Two parameters of key importance to the Domar-Musgrave effect are the risk aversion of agents and the volatility of entrepreneurial productivity. As risk aversion or volatility increases from a low level, with tax rates varying accordingly to maintain optimality, the natural borrowing constraint is initially strictly binding on entrepreneurs, then barely binding, then slack. The optimal rate of capital income taxation is zero while the borrowing constraint is strictly binding, then sharply increasing in risk aversion or volatility while the borrowing constraint is barely binding, then more gradually increasing while the borrowing constraint is slack. The sharp increase in the optimal rate of capital income taxation over the barely binding region confounds any attempt to provide a robust calculation of the optimal tax rate using our model. For instance, if we fix the risk aversion parameter equal to 3 as in \cite{Angeletos2007}, then the borrowing constraint is barely binding when volatility is between 0.198 and 0.215. Both values of volatility are in the empirically relevant range: volatility is 0.2 in the preferred calibration in \cite{Angeletos2007}, and is 0.247 in our preferred calibration. We find that the optimal capital income tax rate rises from zero to 18\% as volatility rises from 0.198 to 0.215. The sharp increase in the optimal capital income tax rate over this small range of volatilities can be viewed as a concerted effort by a social planner to maintain full leverage of physical capital through the Domar-Musgrave effect. We observe a similar phenomenon when varying the risk aversion parameter while holding volatility constant. The qualitatively different behavior of the optimal tax rates over the strictly binding, barely binding and slack regions of the risk aversion and volatility parameter spaces flows on to equilibrium prices (in our model, the interest rate and wage), which also behave differently in these three regions, sometimes in unexpected ways. For instance, holding risk aversion constant at 3 and varying volatility, the equilibrium wage is increasing in volatility while the borrowing constraint is strictly binding or slack, but is decreasing over the intermediate range of volatilities where the borrowing constraint is barely binding.
	
	While the optimal rate of capital income taxation in our model is highly sensitive to model parameters, the optimal rate of labor income taxation is not: it is zero. Intuitively, it is optimal to not tax labor income because taxing consumption is nondistortionary (i.e., it does not affect investment decisions) while also being more progressive than labor income taxation, as the latter fails to generate significant revenue from wealthy agents whose income is primarily derived from capital. The more surprising aspect of our analysis is that it may be optimal to tax capital income in addition to taxing consumption. Indeed, \cite{Coleman2000} is the only prior study we are aware of which finds that it may be optimal to tax capital income when choosing constant rates of taxation on labor income, capital income and consumption. There the optimal constant rates of labor and capital income taxation are zero and 2\% respectively.
	
	Perhaps the most restrictive aspect of our model is the lack of idiosyncratic variation in labor income. All labor earns a fixed wage; it is only capital income which is heterogeneous. This strong assumption preserves analytical tractability and is what allows us to use the results in \cite{BeareToda2022} to derive explicit formulae for excess demand functions and other quantities of interest. With idiosyncratic variation in labor income, the wealth process for a succession of agents would no longer be a Markov multiplicative process with reset, and so the results in \cite{BeareToda2022} would not apply. It is reasonable to wonder whether the optimality of zero labor income taxation in our model is entirely driven by the lack of labor income heterogeneity. Past literature indicates that this is not the case. In particular, \cite{Imrohoroglu1998} presents an analysis of the optimal choice of flat taxes on labor income, capital income and consumption in a model in which there is idiosyncratic variation in labor income but not in the return to capital. This model is, in a sense, the polar opposite of the model considered here. The optimal rate of labor income taxation in \cite{Imrohoroglu1998} is zero, just as it is in our model. On the other hand, the optimal rate of capital income taxation in \cite{Imrohoroglu1998} is also zero, whereas in our model it may be positive. The difference is explained by the Domar-Musgrave effect, which is absent in the model considered in \cite{Imrohoroglu1998} due to the lack of idiosyncratic variation in the return to capital.
	
	Using our preferred model calibration for the United States, we calculate the rates of capital income and consumption taxation which maximize welfare in stationary equilibrium while preserving current tax revenue to be 24\% and 31\% respectively. Eliminating the labor income tax and setting the rates of capital income and consumption taxation equal to these values is calculated to increase welfare in stationary equilibrium by same amount as a 6.6\% permanent increase in the consumption of all agents. The actual increase in aggregate consumption is 4.3\%, but the gain is skewed toward less wealthy agents, with the aggregate consumption of workers rising by 5.7\% and the aggregate consumption of entrepreneurs declining by 2.2\%. While the optimal rates of consumption taxation and particularly capital income taxation are sensitive to model parameters, the general conclusion that replacing the bulk of labor income tax revenue with consumption tax revenue generates an increase in the stationary equilibrium level of welfare is not. Moreover, the magnitude of the increase is large. In a closely related context involving changes in taxation policy, \cite{Lucas1990} refers to a projected increase in aggregate consumption of 7\% as ``the largest genuinely free lunch I have seen in 25 years in this business'', and ``about twice the welfare gain that I have elsewhere estimated would result from eliminating a 10\% inflation''.
	
	Our calculation of a 6.6\% increase in welfare pertains to the welfare level obtained after the economy has adjusted to the new stationary equilibrium under the optimal tax rates. It does not pertain to the immediate or short-term effect on welfare of switching to the optimal tax rates. It is conceivable that the transition to the new stationary equilibrium, which must necessarily involve a period of depressed consumption while capital is accumulated, may be too painful to justify the long-term benefits. To address this matter we calculate the path along which our model economy transitions to the new stationary equilibrium following a change to the optimal tax rates. We find that aggregate consumption drops by 2.3\% for workers and by 10\% for entrepreneurs immediately following the change in tax rates. In subsequent years, the aggregate consumption of both groups rises. The aggregate consumption of workers surpasses its initial level within six years, while the aggregate consumption of entrepreneurs never fully recovers. The aggregate consumption of all agents surpasses its initial level within ten years. The interest rate is elevated during the period of depressed consumption, inducing greater saving by workers which is transferred to entrepreneurs as debt and used to increase the capital stock.
	
	To assess the political viability of changing to the optimal tax rates, we calculate the proportion of agents whose utility rises immediately following such a change. This utility incorporates an agent's current consumption as well as their anticipation of future consumption. We find that 86\% of agents experience a rise in utility immediately following the implementation of optimal tax rates. The percentage is 93\% among workers and 26\% among entrepreneurs. We take this to mean that a large majority of all agents, but only a minority of entrepreneurs, view the temporary period of depressed aggregate consumption to be worth the long-term benefits.
	
	Our research builds on a literature on optimal taxation too voluminous to be effectively summarized here; see \cite{BastaniWaldenstrom2020} for a partial survey focusing on the taxation of capital. Two prominent early theoretical contributions are \cite{Chamley1986} and \cite{Judd1985}. The details of the models analyzed in these articles differ, but share the general setting of an economy populated by infinitely lived agents with perfect foresight. Time paths of labor and capital income tax rates are chosen to maximize welfare subject to generating a fixed amount of revenue. The central finding is that the welfare-maximizing path of capital income tax rates decreases to zero over time. A pointed critique put forward recently in \cite{StraubWerning2020} has disputed technical aspects of the Chamley-Judd result, while acknowledging that the issues raised do not apply to a number of other studies establishing the optimality of long-run zero capital income taxation in variations of the Chamley-Judd framework.
	
	Prominent economists have interpreted the Chamley-Judd result as a justification for eliminating capital income taxation; see, for instance, \cite{AtkesonChariKehoe1999}, helpfully titled ``Taxing capital income: A bad idea''. Others have taken a more skeptical view, with \cite{ConesaKitaoKrueger2009} having the equally helpful title ``Taxing capital? Not a bad idea after all!'' and \cite{BassettoCui2024} arguing that, in the presence of financial frictions, it may be optimal to levy a positive long-run rate of capital income taxation as a means of financing government lending to liquidity-constrained entrepreneurs. \cite{Coleman2000} presents a nuanced perspective, drawing attention to the fact that while the optimal path of capital income tax rates in the Chamley-Judd models is zero in the long-run, it can initially be very high. For instance, in an example with additively separable preferences discussed in \cite{Chamley1986}, the optimal initial rate of capital income taxation is 100\%, and remains at 100\% until some fixed time at which it falls to zero. The period of confiscatory taxation lasts for several years in the numerical calibrations considered in \cite{Coleman2000}. \cite{Lucas1990} observes that a path of capital income tax rates which is initially very high and then falls to zero may be regarded as ``imitating a capital levy on the initial stock'', and in this sense resembles a tax on initial wealth. It is natural to ask whether a consumption tax may provide a simpler alternative to the peculiar path of capital income tax rates which is optimal in the Chamley-Judd framework. This is essentially the question addressed in \cite{Coleman2000}; the answer supplied is, in part, that deriving all revenue from a constant rate of consumption taxation achieves nearly the same welfare as the optimal paths of labor income, capital income and consumption tax rates. Though the modeling framework used in this article is very different, we too find that deriving all revenue from a constant rate of consumption taxation achieves a welfare level that is nearly optimal.
	
	The lack of idiosyncratic variation in the return to capital in the Chamley-Judd framework precludes the possibility of capital income taxation encouraging capital accumulation through the Domar-Musgrave effect. The same is true of other models, such as those in \cite{Aiyagari1995}, \cite{Imrohoroglu1998} and \cite{ConesaKitaoKrueger2009}, in which agents face uncertainty about their future labor income but not about future returns to capital. Within the realm of heterogeneous agent macroeconomics there are relatively few studies of optimal taxation in which the Domar-Musgrave effect plays a significant role. One such study is \cite{Panousi2010}, already mentioned above. Others include \cite{PanousiReis2012,PanousiReis2021}, \cite{BoadwaySpiritus2021} and \cite{GerritsenJacobsRusuSpiritus2020}. These studies all concern the optimal choice of labor and/or capital income tax rates in settings where future returns to capital are uncertain, and find that it is optimal to tax capital income, at least when uncertainty is sufficiently strong. They do not pertain to a setting where a consumption tax is available, as in this article. Past literature shows that the availability of a consumption tax matters a great deal for whether capital income should be taxed. In \cite{Imrohoroglu1998} and \cite{Coleman2000}, the optimal constant rate of capital income taxation is zero or close to zero if a consumption tax is available, but is otherwise substantially greater than zero. We show in this article that the Domar-Musgrave effect can lead the optimal capital income tax rate to be substantially greater than zero even if a consumption tax is available.
	
	\section{Model}\label{sec:model}
	
	Our model extends the one in \cite{Angeletos2007} and is similar to the one in \cite{Panousi2010}. It extends the former by introducing taxation, Markov-switching entrepreneurial ability, and random mortality. The characterization of the stationary distribution of wealth, excess demand equations and equilibrium tax revenue provided in Sections \ref{sec:aggregates}, \ref{sec:equilibrium} and \ref{sec:tax} goes beyond what is provided in the two studies just cited and is obtained by applying results established in \cite{BeareToda2022}. We maintain a discrete-time setting as in \cite{Angeletos2007} rather than the continuous-time setting of \cite{Panousi2010}, but note that \cite{BeareSeoToda2022} provides continuous-time analogs to the discrete-time results in \cite{BeareToda2022} which could be applied in the present context if a continuous-time treatment was preferred. 
	
	\subsection{Agents}\label{sec:agents}
	
	Our model economy is populated by a unit mass of agents. Time is divided into discrete periods indexed by $t\in\mathbb Z_+\coloneqq\{0,1,2,\dots\}$. The decision problems faced by agents in each period $t$ are distinguished by two state variables: their wealth $W_t$, which comprises privately owned capital and risk-free bond holdings, and an entrepreneurial ability state $J_t$ taking values in a finite set $\cN=\{1,\dots,N\}$. An agent's ability state $J_t$ evolves exogenously as a homogeneous Markov chain with irreducible transition probability matrix $\Pi=(\pi_{nn'})$. It affects their productivity when engaged in private enterprise, as described in Section \ref{sec:budget}.
	
	We adopt the perpetual youth framework introduced in \cite{Yaari1965} and \cite{Blanchard1985} in which, when transitioning between time periods, an agent survives with probability $\upsilon\in (0,1)$ and perishes with probability $1-\upsilon$. Actuarially fair life insurance companies trade annuities in exchange for ownership of an agent's wealth or debt upon mortality. Thus, an agent's wealth or debt carried over from one period to the next is multiplied by $1/\upsilon$ if they survive, or is reduced to zero if they perish. When an agent perishes, they are replaced with a new agent endowed with zero wealth and an ability state drawn from a probability distribution $\varpi$ on $\cN$. Mortality occurs independently of all other variables in the model.
	
	\subsection{Production, wealth, and budget constraint}\label{sec:budget}
	
	An agent commences period $t$ with ability state $J_t$ as well as the physical capital $K_t$ and risk-free savings or debt $B_t$ carried over from the previous period. (A new agent is endowed with no resources.) The agent then hires labor $L_t$ at time-invariant price (wage) $\omega$ to operate their private enterprise and produce $F_{J_t}(K_t,L_t)$ units of a consumption good, where $F_n:[0,\infty)^2\to[0,\infty)$ is the production function in ability state $n$.
	
	We allow there to be one ability state (say, $n=1$) for which $F_n$ is identically equal to zero. An agent with this ability state may be understood to be a pure worker. In all other ability states, $F_n$ is assumed to be continuous, nonnegative homogeneous of degree one (i.e., constant-returns-to-scale), strictly concave in the second argument when the first argument is positive (i.e., diminishing marginal returns to labor), and to satisfy $F_n(0,\ell)=0$ for all $\ell\ge 0$ and the Inada condition $\lim_{\ell \to \infty}F_n(1,\ell)/\ell=0$. We maintain these conditions on $F_n$ throughout our analysis. A typical parametrization satisfying them, which we use in our numerical calibration in Section \ref{sec:numerical}, is the Cobb-Douglas production function $F_n(k,\ell)=A_nk^\alpha \ell^{1-\alpha}$, where $A_n\geq0$ is called total factor productivity and $\alpha\in (0,1)$ is called the output elasticity of capital.
	
	As in \citet{Angeletos2007}, we treat capital and the consumption good as interchangeable, and normalize their price to one. We suppose that a flat tax rate of $\tau_\mathrm{K}\in[0,1)$ is applied to the profits from production; that is, $F_{J_t}(K_t,L_t)-\delta K_t-\omega L_t$, the productive output minus the costs of capital depreciation and hired labor. Bond holdings are measured in units of the consumption good, and generate a time-invariant rate of interest. We suppose that interest earned from bonds is taxed at the flat rate $\tau_\mathrm{K}$, and denote by $R$ the post-tax gross rate of interest.\footnote{Alternatively, it can be understood that interest earned from bonds is exempt from capital income taxation, and $R$ is the gross rate of interest without taxation. No revenue is generated from the taxation of interest in equilibrium due to bond market clearing; see Section \ref{sec:equilibrium}.} Thus we define the wealth of an agent after production by
	\begin{equation}\label{eq:initW}
		W_t=K_t+RB_t+(1-\tau_\mathrm{K})[F_{J_t}(K_t,L_t)-\delta K_t-\omega L_t].
	\end{equation}
	
	All agents are assumed to supply one unit of labor inelastically each period, earning pre-tax wage $\omega$. Wages are taxed at a flat rate $\tau_\mathrm{L}\in[0,1)$, so the post-tax wage is $(1-\tau_\mathrm{L})\omega$. Agents choose how to divide their wealth and labor income between current consumption $C_t$ as well as the physical capital $K_{t+1}$ and risk-free bonds $B_{t+1}$ to be held at the beginning of the next period. Consumption is taxed at a flat rate $\tau_\mathrm{C}\in[0,\infty)$. The budget constraint in period $t$ is thus
	\begin{equation}\label{eq:budget}
		(1+\tau_\mathrm{C})C_t+\upsilon K_{t+1}+\upsilon B_{t+1}=W_t+(1-\tau_\mathrm{L})\omega,
	\end{equation}
	where wealth $W_t$ is as in \eqref{eq:initW}, and we recall that the role played by life insurance companies has the effect of multiplying a surviving agent's capital and bond holdings by $1/\upsilon$ each period. The agent is required to choose $C_t\ge 0$ and $K_{t+1}\ge 0$. Bond holdings $B_{t+1}$ may be positive, zero or negative, subject to a natural borrowing constraint. Specifically, we define the borrowing limit $\ubar{b}\le 0$ such that the agent is able to roll over debt indefinitely without engaging in private enterprise or consumption. Combining \eqref{eq:initW} and \eqref{eq:budget} together with $C_t=K_t=K_{t+1}=L_t=0$ and $B_t=B_{t+1}=\ubar{b}$, we see that $\upsilon \ubar{b}=R\ubar{b}+(1-\tau_\mathrm{L})\omega$, so that the borrowing limit $\ubar{b}$ satisfies
	\begin{equation}
		\ubar{b}=-\frac{(1-\tau_\mathrm{L})\omega}{R-\upsilon},\label{eq:borrowlimit}
	\end{equation}
	where we require $R>\upsilon$ for $\ubar{b}$ to be well-defined. We maintain this condition on $R$ throughout our analysis.
	
	\subsection{Preferences}\label{sec:preferences}
	
	As in \citet{Angeletos2007}, agents are assumed to have Epstein-Zin-Weil preferences with discount factor $\beta\in(0,1)$, unit elasticity of intertemporal substitution,\footnote{It is argued convincingly in \citet[p.~14]{Angeletos2007} that the unit case is the relevant one if we are concerned primarily with the behavior of wealthier agents. Other choices of the intertemporal elasticity of substitution greatly complicate the mathematics to follow.} and relative risk aversion $\gamma>0$. Such preferences involve a recursive formulation of utility in which
	\begin{equation}
		U_t=\exp((1-\beta)\log C_t+\beta\log\mu_t(U_{t+1})).\label{eq:EZ}
	\end{equation}
	Here, $U_t>0$ is the agent's utility in period $t$ depending on their current and (uncertain) future consumption, and the quantity $\mu_t(U_{t+1})$ is the Kreps-Porteus certainty equivalent of $U_{t+1}$, given by
	\begin{equation*}
		\mu_t(U_{t+1})=\nu_\gamma^{-1}(\E(\nu_\gamma(U_{t+1})\mid J_t,W_t)),
	\end{equation*}
	where $\nu_\gamma:(0,\infty)\to\R$ is the Box-Cox transformation
	\begin{equation*}
		\nu_\gamma(c)=\begin{cases*}
			\frac{c^{1-\gamma}-1}{1-\gamma}& if $\gamma\neq 1$,\\
			\log c& if $\gamma=1$.
		\end{cases*}
	\end{equation*}
	The discount factor $\beta$ should be understood to incorporate both the agent's preference for current over future consumption, and their awareness of the risk of mortality. For further discussion of recursive utility and Epstein-Zin-Weil preferences we refer the reader to \citet{KrepsPorteus1978}, \citet{EpsteinZin1989} and \citet{Weil1989}.
	
	\subsection{Optimal decision rules}\label{sec:optimal}
	
	Agents choose the labor input for their private enterprise, and choose how to allocate their wealth after production between consumption, physical capital and bonds. It will be convenient to first solve the labor choice problem, which is straightforward, and then solve the more complicated wealth allocation problem. Noting that utility is monotone in consumption and that the agent can choose hired labor $L_t$ after observing the ability state $J_t$ and taking the installed capital $K_t$ as given, it is clear from \eqref{eq:initW} that the agent chooses hired labor $L_t\ge 0$ to maximize the pre-tax profit $F_{J_t}(K_t,L_t)-\delta K_t-\omega L_t$. For each $n\in \cN$ and $\omega>0$ we define
	\begin{subequations}\label{eq:relln}
		\begin{align}
			r_n(\omega)&=(1-\tau_\mathrm{K})\max_{\ell\ge 0}(F_n(1,\ell)-\delta-\omega\ell),\label{eq:rn}\\
			\ell_n(\omega)&=\argmax_{\ell\ge 0}(F_n(1,\ell)-\delta-\omega\ell),\label{eq:elln}
		\end{align}
	\end{subequations}
	noting that the conditions imposed on $F_n$ in Section \ref{sec:budget} (specifically, diminishing marginal returns to labor and the Inada condition) ensure that $F_n(1,\ell)-\omega\ell$ is uniquely maximized by some $\ell\geq0$. Using the assumed homogeneity of $F_n$, the wealth \eqref{eq:initW} maximized over hired labor $L_t\ge 0$ then becomes
	\begin{equation}
		W_t=(1+r_{J_t}(\omega))K_t+RB_t,\label{eq:maxW}
	\end{equation}
	and the optimal labor input is $L_t=\ell_{J_t}(\omega)K_t$. Combining \eqref{eq:budget} and \eqref{eq:maxW} gives
	\begin{equation*}
		(1+\tau_\mathrm{C})C_t+\upsilon K_{t+1}+\upsilon B_{t+1}=(1+r_{J_t}(\omega))K_t+RB_t+(1-\tau_\mathrm{L})\omega.
	\end{equation*}
	Subtracting $\upsilon\ubar{b}=R\ubar{b}+(1-\tau_\mathrm{L})\omega$ from both sides and using \eqref{eq:borrowlimit}, we obtain
	\begin{equation}
		(1+\tau_\mathrm{C})C_t+\upsilon K_{t+1}+\upsilon (B_{t+1}-\ubar{b})=(1+r_{J_t}(\omega))K_t+R(B_t-\ubar{b}). \label{eq:budget2}
	\end{equation}
	Equation \eqref{eq:budget2} may be viewed as the budget constraint which applies to the agent's wealth allocation problem, labor having already been optimally chosen.
	
	To solve the wealth allocation problem it will be useful to introduce further notation. We define the human wealth (present discounted value of future labor income taking into account mortality)
	\begin{equation}
		h=\sum_{t=0}^\infty (\upsilon/R)^t(1-\tau_\mathrm{L})\omega=\frac{(1-\tau_\mathrm{L})\omega}{1-\upsilon/R}=-R\ubar{b},\label{eq:humanwealth}
	\end{equation}
	the total wealth (financial wealth plus human wealth)
	\begin{equation}
		S_t= (1+r_{J_t}(\omega))K_t+R(B_t-\ubar{b})=W_t+h,\label{eq:S_t}
	\end{equation}
	and the fraction of post-consumption total wealth allocated to physical capital
	\begin{equation}
		\theta_t= \frac{\upsilon K_{t+1}}{S_t-(1+\tau_\mathrm{C})C_t}\ge 0.\label{eq:frack}
	\end{equation}
	Using \eqref{eq:budget2}, \eqref{eq:S_t}, and the borrowing limit $B_{t+1}\ge \ubar{b}$, it follows from the definition of $\theta_t$ in \eqref{eq:frack} that
	\begin{equation}
		1-\theta_t=\frac{\upsilon(B_{t+1}-\ubar{b})}{S_t-(1+\tau_\mathrm{C})C_t}\ge 0.\label{eq:fracb}
	\end{equation}
	We must therefore have $\theta_t\in [0,1]$. Combining \eqref{eq:budget2}--\eqref{eq:fracb}, we can compactly write the budget constraint for the wealth allocation problem as
	\begin{equation}
		S_{t+1}=R_{J_{t+1}}(\theta_t)(S_t-(1+\tau_\mathrm{C})C_t),\label{eq:budget3}
	\end{equation}
	where we define the gross return on total wealth
	\begin{equation*}
		R_n(\theta)= \frac{1}{\upsilon}((1+r_n(\omega))\theta+R(1-\theta)),
	\end{equation*}
	suppressing its dependence on $R$ and $\omega$. We will return to \eqref{eq:budget3} in Section \ref{sec:aggregates} when our focus turns to the distribution of wealth.
	
	The wealth allocation problem solved by an agent with total wealth $S_t$ in period $t$ may be viewed as a maximization over two variables: current consumption $C_t\in[0,S_t/(1+\tau_\mathrm{C})]$ and the portfolio weight $\theta_t\in[0,1]$. Given these two variables, we may recover the agent's choice of $K_{t+1}$ from \eqref{eq:frack} and their choice of $B_{t+1}$ from \eqref{eq:fracb}, these choices automatically satisfying the budget constraint \eqref{eq:budget2}. Let $V_n^*(s)$ be the value function for the wealth allocation problem: the utility $U_t$ achieved by an optimally behaving agent with $J_t=n$ and $S_t=s$. In view of the utility recursion \eqref{eq:EZ} and Bellman's principle of optimality, the value function $V_n^*(s)$ solves the Bellman equation defined by
	\begin{equation}\label{eq:vna}
		V_n(s)=\max_{\substack{c\in[0,s/(1+\tau_\mathrm{C})]\\ \theta\in [0,1]}}\exp\left((1-\beta)\log c+\beta\log \mu_n\left(V_{n'}(R_{n'}(\theta)(s-(1+\tau_\mathrm{C})c))\right)\right),
	\end{equation}
	where $n'=J_{t+1}$, and $\mu_n(\cdot)=\nu_\gamma^{-1}\left(\E(\nu_\gamma(\cdot)\mid J_t=n)\right)$ is the Kreps-Porteus certainty equivalent conditional on $J_t=n$.
	
	We will shortly state a result, Lemma \ref{lem:optimal}, solving the maximization on the right-hand side of \eqref{eq:vna} when the candidate value function is of the form $V_n(s)=a_ns$, where $a_1,\dots,a_N$ are arbitrary positive constants. To this end, for each $n\in\cN$ and $a=(a_1,\dots,a_N)\gg 0$, we define the function $g_n(\cdot;a):[0,1]\to \R$ by
	\begin{equation}\label{eq:g}
		g_n(\theta;a)=\sum_{n'=1}^N\pi_{nn'}\nu_\gamma(a_{n'}R_{n'}(\theta)).
	\end{equation}
	Noting that $\nu_\gamma$ is continuous and strictly concave and $\theta\mapsto R_{n'}(\theta)$ is affine, we see that $g_n(\cdot;a)$ is continuous and concave, and strictly so unless $1+r_{J_{t+1}}(\omega)=R$ almost surely conditional on $J_t=n$. Hence (generically) there exists a unique value of $\theta \in [0,1]$ at which $g_n(\theta;a)$ is maximized, which we denote by $\theta_n(a)$. (If $1+r_{J_{t+1}}(\omega)=R$ almost surely conditional on $J_t=n$, then any $\theta\in [0,1]$ is optimal.) Define $\kappa_n(a)=\nu_\gamma^{-1}(g_n(\theta_n(a);a))$.
	
	\begin{lem}\label{lem:optimal}
		Let $\theta_n(a)$ be a maximizer of \eqref{eq:g}. The maximum on the right-hand side of \eqref{eq:vna} with candidate value function $V_n(s)=a_ns$ is achieved by setting $\theta=\theta_n(a)$ and
		\begin{equation*}
			c=\frac{1-\beta}{1+\tau_\mathrm{C}}s.
		\end{equation*}
		The maximum achieved is equal to
		\begin{equation*}
			\left(\frac{1-\beta}{1+\tau_\mathrm{C}}\right)^{1-\beta}(\beta\kappa_n(a))^\beta s.
		\end{equation*}
	\end{lem}
	A proof of Lemma \ref{lem:optimal}, and proofs of the other numbered mathematical statements in this section, are provided in Appendix \ref{sec:proofwealth}. It is clear from Lemma \ref{lem:optimal} that if a solution to the Bellman equation \eqref{eq:vna} takes the form $V_n(s)=a_ns$, then $a$ satisfies the system of nonlinear equations
	\begin{equation}\label{eq:nlsystem}
		a_n=\left(\frac{1-\beta}{1+\tau_\mathrm{C}}\right)^{1-\beta}(\beta\kappa_n(a))^\beta,\quad n\in\cN.
	\end{equation}
	The following result establishes that there is a unique solution $a=a^\ast$ to \eqref{eq:nlsystem} and that the value function $V_n^\ast(s)=a_n^\ast s$ solves the Bellman equation \eqref{eq:vna}. It also characterizes the agent's optimal decision rules.
	
	\begin{prop}\label{prop:optrule}
		The Bellman equation \eqref{eq:vna} is solved by the value function $V_n^*(s)=a_n^*s$, where $a^*=(a_1^*,\dots,a_N^*)$ uniquely solves \eqref{eq:nlsystem}. Letting $\theta_n^\ast=\theta_n(a^\ast)$, the optimal choices of consumption, capital, labor, and bonds corresponding to $V_n^\ast(s)$ are
		\begin{subequations}\label{eq:optrule}
			\begin{align}
				C_n^*(s)&=\frac{1-\beta}{1+\tau_\mathrm{C}}s,\label{eq:crule}\\
				K_n^*(s)&=\frac{\beta}{\upsilon}\theta_n^* s,\label{eq:krule}\\
				L_n^*(s)&=\frac{\beta}{\upsilon}\theta_n^*\ell_n(\omega)s,\label{eq:lrule}\\
				B_n^*(s)&=-\frac{h}{R}+\frac{\beta}{\upsilon}(1-\theta_n^*) s.\label{eq:brule}
			\end{align}
		\end{subequations}
	\end{prop}
	The solution to the Bellman equation provided by Proposition \ref{prop:optrule} is unique within a wide class of candidate value functions. See Proposition \ref{pro:uniqueV} in Appendix \ref{sec:proofwealth} for details.
	
	\subsection{Stationary distribution of wealth}\label{sec:aggregates}
	
	Substituting the optimal choices $C_{J_t}^*(S_t)$ and $\theta^\ast_{J_t}$ into the budget constraint \eqref{eq:budget3}, we obtain the following law of motion for the total wealth of an optimizing agent:
	\begin{equation*}
		S_{t+1}=\beta R_{J_{t+1}}(\theta_{J_t}^*)S_t=G_{J_tJ_{t+1}}S_t,
	\end{equation*}
	where $G_{nn'}\coloneqq \beta R_{n'}(\theta_n^*)$ is the gross growth rate conditional on transitioning from $J_t=n$ to $J_{t+1}=n'$. The total wealth $S_t$ of an optimizing agent is thus formed as an accumulation of multiplicative shocks, with each shock determined by the current and previous value of the exogenous Markov switching ability state.
	
	As in \citet{Yaari1965} and \cite{Blanchard1985}, the agents in our model are assumed to perish with a fixed probability $1-\upsilon\in(0,1)$ each period, being replaced with a new agent with zero financial wealth and ability state drawn from the distribution $\varpi$. With a slight abuse of notation, in what follows we let $S_t$ and $J_t$ denote the total wealth and ability state of a \emph{succession} of agents. When an agent perishes between periods $t$ and $t+1$ and is replaced with a new agent, because a new agent is endowed only with their human wealth $h$, $S_{t+1}$ is reset to $h$ and $J_{t+1}$ is drawn from the distribution $\varpi$. The sequence of pairs $(S_t,J_t)_{t\in\mathbb Z_+}$ is thus a \emph{Markov multiplicative process with reset} as defined in \citet{BeareToda2022}, with the obvious modification that the reset value is $h$ instead of one.
	
	We can apply the results in \citet{BeareToda2022} on Markov multiplicative processes with reset to characterize the distribution of wealth in our model economy. Proposition 3 in that article implies the existence of a unique distribution for $(S_0,J_0)$ such that $(S_t,J_t)_{t\in\mathbb Z_+}$ is stationary; we call the time-invariant distribution of $(S_t,J_t)$ under this stationary initialization the \emph{stationary joint distribution of wealth and ability}, and call the time-invariant distribution of $S_t$ the \emph{stationary distribution of wealth}. The time-invariant distribution of $J_t$ can be represented by an $N\times1$ vector $p$ whose $n$th entry $p_n$ is the probability that $J_t=n$ under stationarity. Taking into account reset, this is the unique stationary distribution corresponding to the irreducible transition probability matrix $\upsilon\Pi+(1-\upsilon)1_N\varpi^\top$, where $1_N$ denotes an $N\times1$ vector of ones.
	
	Proposition \ref{prop:stationary}, to be stated momentarily, provides a formula for the Mellin transform of the stationary distribution of wealth conditional on the ability state, and establishes (under a mild regularity condition) that the right tail of the stationary distribution of wealth is Pareto (i.e., exhibits a power law) with a certain rate of decay. It is proved using Theorem 1, Lemma 2 and Proposition 3 in \citet{BeareToda2022}. Our statement of Proposition \ref{prop:stationary} requires some additional notation. We let $\A(z)$ denote an $N\times N$ matrix $\A(z)$ depending on a complex variable $z$, with $(n,n')$-entry equal to $\upsilon\pi_{nn'}G_{nn'}^z$. We let $\mathcal I_-$ denote the set of all real $z$ such that $\rho(\A(z))$, the spectral radius of $\A(z)$, is less than one. Proposition 1 in \citet{BeareToda2022} implies that $\rho(\A(z))$ is a convex function of real $z$ satisfying $\rho(\A(0))=\upsilon<1$, so the set $\mathcal I_-$ is convex and contains both positive and negative values. We let $\mathrm{I}$ denote the $N\times N$ identity matrix, and let $e^{(n)}$ denote the $N\times1$ vector with $n$th entry equal to one and all other entries equal to zero.
	\begin{prop}\label{prop:stationary}
		For each complex $z$ with real part belonging to $\mathcal I_-$, the matrix $\mathrm{I}-\A(z)$ is invertible, and a random draw $(S,J)$ from the stationary joint distribution of wealth and ability satisfies
		\begin{align}\label{eq:conditionalwealthMellin}
			\E(S^z\mid J=n)&=(1-\upsilon)p_n^{-1}h^z\varpi^\top(\mathrm{I}-\A(z))^{-1}e^{(n)}
		\end{align}
		for each $n\in\mathcal N$ with $p_n>0$, and
		\begin{align}\label{eq:wealthMellin}
			\E(S^z)&=(1-\upsilon)h^z\varpi^\top(\mathrm{I}-\A(z))^{-1}1_N.		
		\end{align}
		Further, if the equation $\rho(\A(z))=1$ admits a unique positive solution $z=\zeta$, then
		\begin{align}\label{eq:Pareto}
			\lim_{s\to\infty}\frac{\log\mathrm{P}(S>s)}{\log s}&=-\zeta,
		\end{align}
		meaning that the right tail of the stationary distribution of wealth is Pareto with decay rate $\zeta$.
	\end{prop}
	
	A mild sufficient condition for the equation $\rho(\A(z))=1$ to admit a unique positive solution, so that the right tail of the stationary distribution of wealth is Pareto, is that we have $\pi_{nn}>0$ and $G_{nn}>1$ for some state $n$. This is a consequence of Proposition 2 in \citet{BeareToda2022}. Note that financial wealth $W_t$ differs from total wealth $S_t$ by the constant value $h$, so if the right tail of the stationary distribution of (total) wealth is Pareto with decay rate $\zeta$ then the same is true for the stationary distribution of financial wealth.
	
	Equation \eqref{eq:wealthMellin} in Proposition \ref{prop:stationary} facilitates the direct computation of the stationary distribution of wealth without resorting to Monte Carlo simulation. Since \eqref{eq:wealthMellin} is valid in particular for all imaginary $z$, it provides the Fourier transform (i.e., characteristic function) of the distribution of $\log S$. The distribution of $\log S$ may be recovered by Fourier inversion and then suitably modified to obtain the distribution of $S$. The conditional distribution of $S$ given $J=n$ may be computed by applying Fourier inversion to \eqref{eq:conditionalwealthMellin} in the same way.
	
	\subsection{Stationary equilibrium}\label{sec:equilibrium}
	
	The stationary joint distribution of wealth and ability depends on two prices: the gross risk-free rate $R$ and the wage $\omega$. We regard these prices as endogenous parameters to be determined by market clearing conditions for the bond market and labor market. We suppose that the risk-free bond is in zero net supply, so that the bond market clears when aggregate demand for bonds is zero. Agents are assumed to supply one unit of labor inelastically each period, so the labor market clears when aggregate demand for labor is one. We therefore say that a given gross risk-free rate $R>\upsilon$ and wage $\omega>0$ constitute a \emph{stationary equilibrium} if a random draw $(S,J)$ from the stationary joint distribution of wealth and ability satisfies the following two conditions:
	\begin{subequations}
		\begin{align}
			\E(B_{J}^*(S))&=0\quad\text{(\emph{bond market clearing})}, \label{eq:bondclear}\\
			\E(L_{J}^*(S))&=1\quad\text{(\emph{labor market clearing})}. \label{eq:laborclear}
		\end{align}
	\end{subequations}
	
	Checking whether a given gross risk-free rate $R$ and wage $\omega$ constitute a stationary equilibrium requires evaluating $\E(B_J^*(S))$ and $\E(L_J^*(S))$, the aggregate demands for bonds and labor. Using Proposition \ref{prop:stationary}, we obtain the following result showing how these aggregate demands may be computed from model parameters. We require the following notation: let $\theta^\ast$ denote the $N\times1$ vector with $n$th entry $\theta^\ast_n$, let $\ell$ denote the $N\times1$ vector with $n$th entry $\ell_n(\omega)$, and let $\odot$ denote the Hadamard (entry-wise) product.
	
	\begin{prop}\label{prop:equilibrium}
		Let $(S,J)$ be a random draw from the stationary joint distribution of wealth and ability. If $\rho(\A(1))<1$, then
		\begin{subequations}
			\begin{align}
				\E(B_{J}^*(S))&=\frac{1-\upsilon}{\upsilon}\beta h\varpi^\top  (\mathrm{I}-\A(1))^{-1}(1_N-\theta^\ast)-\frac{h}{R},\label{eq:aggB}\\
				\E(L_{J}^*(S))&=\frac{1-\upsilon}{\upsilon}\beta h\varpi^\top  (\mathrm{I}-\A(1))^{-1}(\theta^\ast\odot \ell).\label{eq:aggL}
			\end{align}
		\end{subequations}
		Otherwise, $\E(S)=\infty$.
	\end{prop}
	
	The case where $\E(S)=\infty$ may be regarded as practically uninteresting as it carries the interpretation of infinite aggregate wealth, indicating an unsuitable choice of model parameters.
	
	The aggregate demand formulae in Proposition \ref{prop:equilibrium} may be used to search numerically for a gross risk-free rate $R$ and wage $\omega$ constituting a stationary equilibrium. The two equations defining our market clearing conditions are in general nonlinear in $R$ and $\omega$. Note that $h$, $\A(1)$, $\theta^\ast$ and $\ell$ depend on $R$ and/or $\omega$, though we have suppressed this in our notation. Care should be taken in conducting a numerical search for stationary equilibrium, as conditions ensuring the existence or uniqueness of equilibrium are not available.
	
	\subsection{Welfare}
	
	To identify an optimal taxation mix we require a measure of overall welfare. As in \cite{Imrohoroglu1998} and other studies, we define welfare in terms of the utility of a randomly selected newborn agent. Specifically, we define welfare to be the Kreps-Porteus certainty equivalent of the utility of a randomly selected newborn agent in stationary equilibrium. Since newborn agents are endowed with only their human wealth $h$ and have their ability state drawn from $\varpi$, welfare is equal to $\mu(V_{J}^\ast(h))$, where $\mu(\cdot)=\nu_\gamma^{-1}(\E(\nu_\gamma(\cdot)))$ and $J$ is a random draw from $\varpi$.
	
	The following result provides a simple formula for welfare. The logarithm $\log a^\ast$ and power $(a^\ast)^{1-\gamma}$ should be understood to refer to the entry-wise application of the logarithm or power function to the vector $a^\ast$.
	
	\begin{prop}\label{prop:welfare}
		Let $V_n^\ast(s)=a_n^\ast s$ as in Proposition \ref{prop:optrule}. If $J$ is a random draw from $\varpi$, then
		\begin{align}
			\mu(V_{J}^\ast(h))&=\begin{cases*}
				h\left(\varpi^\top (a^*)^{1-\gamma}\right)^\frac{1}{1-\gamma}& if $\gamma\neq 1$,\\
				h\exp\left(\varpi^\top\log a^\ast\right)& if $\gamma=1$.
			\end{cases*}\label{eq:Wnew}
		\end{align}
	\end{prop}
	
	Note that because the value function $V_n^\ast(s)=a_n^\ast s$ is homogeneous of degree one in total wealth $s$, and agents consume a fixed fraction of their total wealth each period, an $x\%$ increase in welfare may be viewed as equivalent to a permanent $x\%$ increase in consumption.
		
	\subsection{Tax revenue}\label{sec:tax}
	
	To characterize the aggregate tax revenue in stationary equilibrium it is simplest to start with the revenues from labor income and consumption, which are conceptually straightforward. Due to all agents supplying one unit of labor at wage $\omega$, and the excess demand for labor being zero in equilibrium, the aggregate tax revenue from labor income is simply $T_\mathrm{L}\coloneqq\tau_\mathrm{L}\omega$. In view of the optimal consumption rule given in Proposition \ref{prop:optrule}, the consumption tax paid by an agent depends only on their current wealth $s$, and is given by
	\begin{align*}
		T_\mathrm{C}(s)&\coloneqq\frac{\tau_\mathrm{C}}{1+\tau_\mathrm{C}}(1-\beta)s.
	\end{align*}
	The aggregate tax revenue from consumption is therefore $\E(T_\mathrm{C}(S))$, where $S$ is a random draw from the stationary distribution of wealth.
	
	The capital income tax paid by an agent depends on their current ability state as well as their previous wealth and ability states. The determination of aggregate tax revenue from capital income is therefore slightly more complicated. Only surviving agents pay capital income tax, because newborn agents are endowed with no bonds or physical capital. We may disregard capital income tax accrued on bond returns because it is zero in aggregate due to bond market clearing. In view of Proposition \ref{prop:optrule}, a surviving agent with previous ability state $n$ and previous wealth $s$ commands physical capital $(\beta/\upsilon)\theta^\ast_ns$ in the current period. If such an agent has current ability state $n'$ then, recalling \eqref{eq:rn}, we deduce that they earn post-tax profit from production equal to $(\beta/\upsilon)\theta_n^\ast sr_{n'}(\omega)$. Therefore, the expected capital income tax (excluding tax accrued on bond returns) paid by a surviving agent with previous ability state $n$ and previous wealth $s$ is
	\begin{align*}
		T_\mathrm{K}(s,n)&\coloneqq\frac{\tau_\mathrm{K}}{1-\tau_\mathrm{K}}(\beta/\upsilon)\theta^\ast_ns\sum_{n'=1}^N\pi_{nn'}r_{n'}(\omega).
	\end{align*}
	The surviving agents constitute fraction $\upsilon$ of all agents, so the total capital income tax paid by all surviving agents is $\upsilon\E(T_\mathrm{K}(S,J))$, where $(S,J)$ is a random draw from the stationary joint distribution of wealth and ability.
	
	We have deduced that aggregate tax revenue in stationary equilibrium is the sum of $T_\mathrm{L}$, $\E(T_\mathrm{C}(S))$ and $\upsilon\E(T_\mathrm{K}(S,J))$. The following result provides convenient formulae for computing $\E(T_\mathrm{C}(S))$ and $\upsilon\E(T_\mathrm{K}(S,J))$. We require one piece of additional notation: let $r$ denote the $N\times1$ vector with $n$th entry $r_n(\omega)$.
	
	\begin{prop}\label{prop:tax}
		If $(S,J)$ is a random draw from the stationary joint distribution of wealth and ability, and if $\rho(\A(1))<1$, then
		\begin{align*}
			\E(T_\mathrm{C}(S))&=\frac{\tau_\mathrm{C}}{1+\tau_\mathrm{C}}(1-\beta)(1-\upsilon)h\varpi^\top(\mathrm{I}-\A(1))^{-1}1_N,\\
			\upsilon\E(T_\mathrm{K}(S,J))&=\frac{\tau_\mathrm{K}}{1-\tau_\mathrm{K}}\beta(1-\upsilon)h\varpi^\top(\mathrm{I}-\A(1))^{-1}(\Pi r\odot\theta^\ast).
		\end{align*}
	\end{prop}
	
	\section{Numerical calibration}\label{sec:numerical}
	
	\subsection{Choice of parameters}\label{sec:calibration}
	
	To obtain numerical predictions from our model we calibrate it to match key features of the U.S.\ economy, with each period understood to be a year in duration. The calibration mostly follows \citet{Angeletos2007} where possible, although there are additional parameters to specify capturing the taxation and mortality rates and the dynamic behavior of entrepreneurial ability.
	
	Following \citet{Angeletos2007} we assume that all production functions are Cobb-Douglas with capital share parameter $\alpha$ equal to $0.36$, and we set the depreciation rate $\delta$ equal to $0.08$, the discount factor $\beta$ equal to $0.96$, and the relative risk aversion $\gamma$ equal to $3$. We set the survival probability $\upsilon$ equal to $0.975$ so that the average lifespan  of an agent (to be interpreted as the average length of economic life) is $(1-\upsilon)^{-1}=40$ years. We set the labor income tax rate $\tau_\mathrm{L}$ equal to $0.25$, which is the current average rate (net cash transfers) for a single worker in the U.S.\ earning the average wage, as reported in \citet[p.~636]{OECD2023}. Following \cite{AokiNirei2017} we define the capital income tax rate $\tau_\mathrm{K}$ by the equality $1-\tau_\mathrm{K}=(1-\tau^\mathrm{cap})(1-\tau^\mathrm{corp})$, where $\tau^\mathrm{cap}$ and $\tau^\mathrm{corp}$ are the capital gains and corporate tax rates. We set $\tau^\mathrm{cap}$ equal to $0.238$ (the current highest rate for long-term capital gains) and $\tau^\mathrm{corp}$ equal to $0.21$ (the current U.S. rate since the Tax Cuts and Jobs Act of 2017), resulting in $\tau_\mathrm{K}=0.4$. Figure \ref{fig:taxrates} plots U.S.\ historical capital gains,\footnote{\url{https://taxfoundation.org/data/all/federal/federal-capital-gains-tax-collections-historical-data}.} corporate\footnote{\url{https://taxfoundation.org/data/all/federal/historical-corporate-tax-rates-brackets}.} and effective capital income tax rates since 1954, as reported by the Tax Foundation. We consider a positive consumption tax rate $\tau_\mathrm{C}$ in analysis to follow, but for now set it equal to zero.
	
	\begin{figure}
		\centering
		\includegraphics[width=\linewidth]{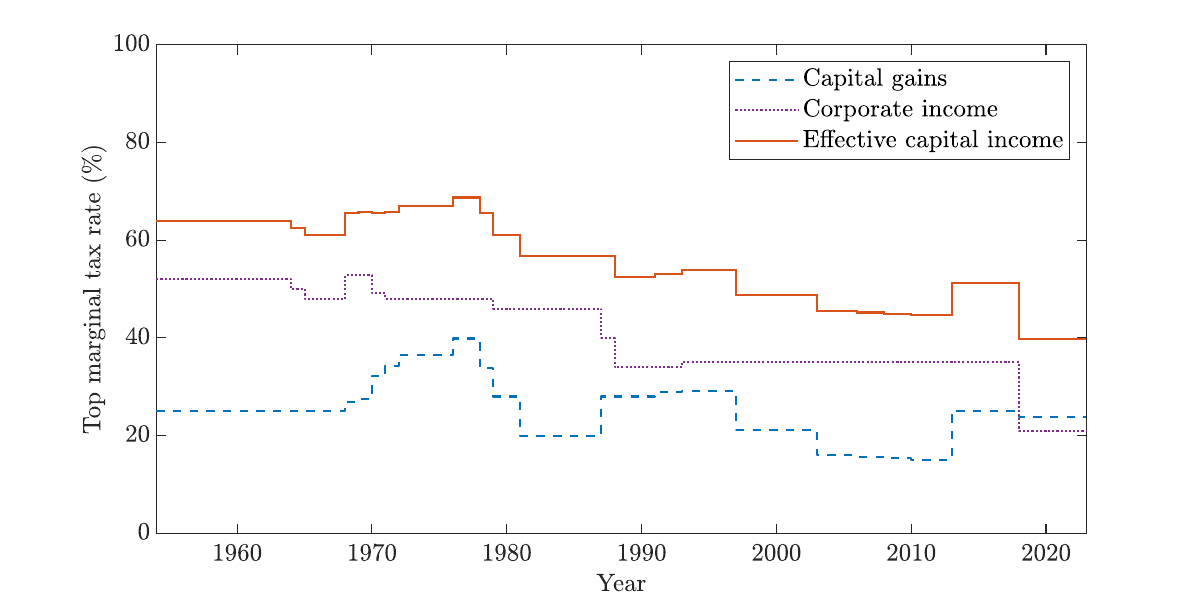}
		\caption{U.S.\ historical capital income tax rates.}\label{fig:taxrates}
	\end{figure}
	
	We model the state-dependent total factor productivity $A_n$ of Cobb-Douglas production as follows. First, we suppose that an agent has zero productive capacity in the first entrepreneurial ability state ($A_1=0$, interpreted as a pure worker) and positive productive capacity in other entrepreneurial ability states ($A_n>0$, interpreted as an entrepreneur). The transition probability matrix for occupation (worker and entrepreneur) is denoted by
	\begin{equation*}
		\Pi_\mathrm{o}=\begin{bmatrix}
			1-\pi_\mathrm{we} & \pi_\mathrm{we}\\
			\pi_\mathrm{ew} & 1-\pi_\mathrm{ew}
		\end{bmatrix},
	\end{equation*}
	where $\pi_\mathrm{we}$ is the transition probability from worker to entrepreneur and $\pi_\mathrm{ew}$ is the transition probability from entrepreneur to worker. We interpret the transition from entrepreneur to worker as firm exit and set $\pi_\mathrm{ew}$ equal to $0.02$ based on \citet[Table 1]{GilchristYankovZakrajsek2009}, where it is documented that the average credit spread for large non-financial firms is 192 basis points (1.92\%). We set $\pi_\mathrm{we}$ equal to $0.0026$ so that the fraction of entrepreneurs in the economy is $\pi_\mathrm{we}/(\pi_\mathrm{we}+\pi_\mathrm{ew})=0.115$, which is the fraction of ``active business owners'' reported in \citet[Table 1]{CagettiDeNardi2006}.
	
	Second, conditional on being an entrepreneur, we assume that productivity can take five values, independently and identically distributed over time. We choose the five values ($A_2,\dots,A_6$) and their probabilities (a $5\times 1$ vector $p_A$) to match specific values of the first four moments of the distribution of entrepreneurial productivity. We normalize the mean to be zero, i.e.\ $\E(\log A)=0$. We set the standard deviation $\sigma$ equal to $0.247$, the skewness equal to $-0.08$, and the kurtosis equal to $6.22$. These numbers are the moments of risky financial asset returns calculated in \citet[Table 3]{FagerengGuisoMalacrinoPistaferri2020} as part of a detailed analysis of Norwegian administrative tax data. We follow \cite{GuvenenKambourovKuruscuOcampoChen2023} and \cite{GerritsenJacobsRusuSpiritus2020} in basing our calibration of U.S.\ asset return heterogeneity on this analysis; no comparably detailed analysis using U.S.\ data is available. The standard deviation of $0.247$ is somewhat larger than the value of $0.2$ used in \cite{Angeletos2007}, but it is noted there that substantial uncertainty surrounds the choice of this parameter. In any case, we devote considerable attention to what happens when $\sigma$ is varied in the analysis to follow, and also to variation in the risk aversion parameter $\gamma$. To match the four specified moments we space $\log(A_2),\dots,\log(A_6)$ evenly between $\pm\sqrt{10}\sigma$ and then choose the vector of probabilities $p_A$ to maximize entropy subject to matching the four moments, consistent with the general procedure described in \cite{TanakaToda2013}.
	
	We complete our specification of the dynamic behavior of entrepreneurial ability by defining the combined transition probability matrix on the state space $\cN=\{1,\dots,N\}$ (with $N=1+5=6$) by the equation
	\begin{equation*}
		\Pi=\begin{bmatrix}
			1-\pi_\mathrm{we} & \pi_\mathrm{we}p_A^\top\\
			\pi_\mathrm{ew}1_{N-1} & (1-\pi_\mathrm{ew})1_{N-1}p_A^\top
		\end{bmatrix}.
	\end{equation*}
	We set the initial productivity distribution $\varpi$ equal to the stationary distribution for $\Pi$. Table \ref{t:param} summarizes the model parameters.
	\begin{table}
		\caption{Model parameters.}\label{t:param}
		\begin{tabular*}{\textwidth}{@{\extracolsep\fill}lcll}
			\toprule
			\textbf{Description} & \textbf{Notation} & \textbf{Value} & \textbf{Source}\\
			\midrule
			Capital share & $\alpha$ & 0.36 & \citet{Angeletos2007} \\
			Capital depreciation & $\delta$ & 0.08 & \citet{Angeletos2007} \\
			Discount factor & $\beta$ & 0.96 & \citet{Angeletos2007} \\
			Relative risk aversion & $\gamma$ & 3 & \citet{Angeletos2007} \\
			Survival probability & $\upsilon$ & 0.975 & Avg.\ 40 year economic life \\
			Capital income tax rate & $\tau_\mathrm{K}$ & 0.398 & Tax Foundation \\
			Labor income tax rate & $\tau_\mathrm{L}$ & 0.248 & \cite{OECD2023} \\
			$\Pr(\text{entrepreneur}\to\text{worker})$ & $\pi_\mathrm{ew}$ & 0.0192 & \citet{GilchristYankovZakrajsek2009} \\
			Fraction of entrepreneurs & $\frac{\pi_\mathrm{we}}{\pi_\mathrm{we}+\pi_\mathrm{ew}}$ & 0.115 & \citet{CagettiDeNardi2006} \\
			Volatility of productivity & $\sigma$ & 0.2473 & \citet{FagerengGuisoMalacrinoPistaferri2020} \\
			Skewness of productivity & -- & -0.08 & \citet{FagerengGuisoMalacrinoPistaferri2020} \\
			Kurtosis of productivity & -- & 6.22 & \citet{FagerengGuisoMalacrinoPistaferri2020} \\
			\bottomrule
		\end{tabular*}
	\end{table}
	
	\subsection{Stationary equilibrium wealth distribution}
	
	We compute the stationary equilibrium by using the aggregate demand formulae in Proposition \ref{prop:equilibrium} to numerically solve the equilibrium conditions \eqref{eq:bondclear} and \eqref{eq:laborclear} for the equilibrium prices $R$ and $\omega$. The equilibrium post-tax risk-free interest rate $R-1$ is 1.7\% (meaning that the pre-tax risk-free interest rate is 2.9\%) and the equilibrium pre-tax wage $\omega$ is $1.27$. We find that the unique positive solution to $\rho(\A(z))=1$ is $z=1.93$. It therefore follows from Proposition \ref{prop:stationary} that the upper tail of the stationary equilibrium distribution of total wealth $S$ (and thus also financial wealth $W$) is Pareto with exponent $\zeta=1.93$.
	
	We use Fourier inversion to compute the stationary equilibrium distribution of wealth in our economy. Specifically, we numerically evaluate the integral in the Gil-Pelaez inversion formula
	\begin{align*}
		\mathrm{Pr}(\log S\leq y)&=\frac{1}{2}-\frac{1}{\pi}\int_0^\infty\operatorname{Im}\left(t^{-1}e^{-ity}\E(S^{it})\right)\mathrm{d}t
	\end{align*}
	using the procedure described in \cite{Witkovsky2016}.\footnote{We use the \textsc{Matlab} routine \texttt{cf2DistGP.m} available at \url{https://github.com/witkovsky/CharFunTool}.} A formula for the Fourier transform $\E(S^{it})$ for the stationary equilibrium distribution of log-wealth can be obtained from Proposition \ref{prop:stationary} by confining the argument $z$ of the Mellin transform $\E(S^{z})$ to the imaginary line. We found that this numerical calculation of probabilities provides much greater accuracy at high wealth levels (up to around the top $10^{-6}$ quantile) than can be achieved using Monte Carlo simulation with comparable runtime. None of the results reported in this article were computed using Monte Carlo simulation.
	
	In Figure \ref{fig:tailProb} we plot the stationary equilibrium distribution of financial wealth computed by Fourier inversion up to a wealth level of $10^4$, which is approximately the top $10^{-6}$ quantile of wealth. Financial wealth $W$ is equal to total wealth $S$ minus human wealth $h$, where we compute $h=22.9$. At wealth levels higher than $10^4$ it becomes burdensome to accurately compute probabilities by Fourier inversion. However, we know from Proposition \ref{prop:stationary} that the right tail of the wealth distribution is Pareto with exponent $\zeta=1.93$. We use this fact to extrapolate our computed probabilities out to a wealth level of $10^6$, which is above the top $10^{-9}$ quantile of wealth. The extrapolation, which is affine in the log-log scale used in Figure \ref{fig:tailProb}, seamlessly extends the probabilities computed by Fourier inversion. In \cite{GouinBonenfantToda2023} the same procedure is used to extrapolate a computed wealth distribution satisfying a known Pareto law, although there the Fourier transform for the distribution of log-wealth is not available in closed form so a method different from Fourier inversion is used to compute the body of the distribution.
	
	\begin{figure}
		\centering
		\begin{subfigure}{0.48\linewidth}
			\includegraphics[width=\linewidth]{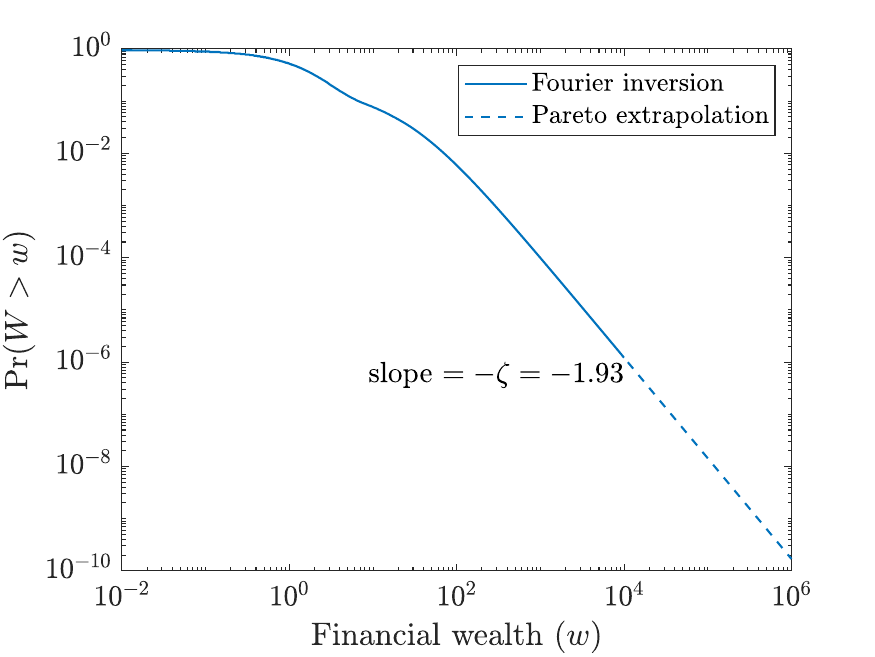}
			\caption{Exceedance probabilities.}\label{fig:tailProb}
		\end{subfigure}
		\begin{subfigure}{0.48\linewidth}
			\includegraphics[width=\linewidth]{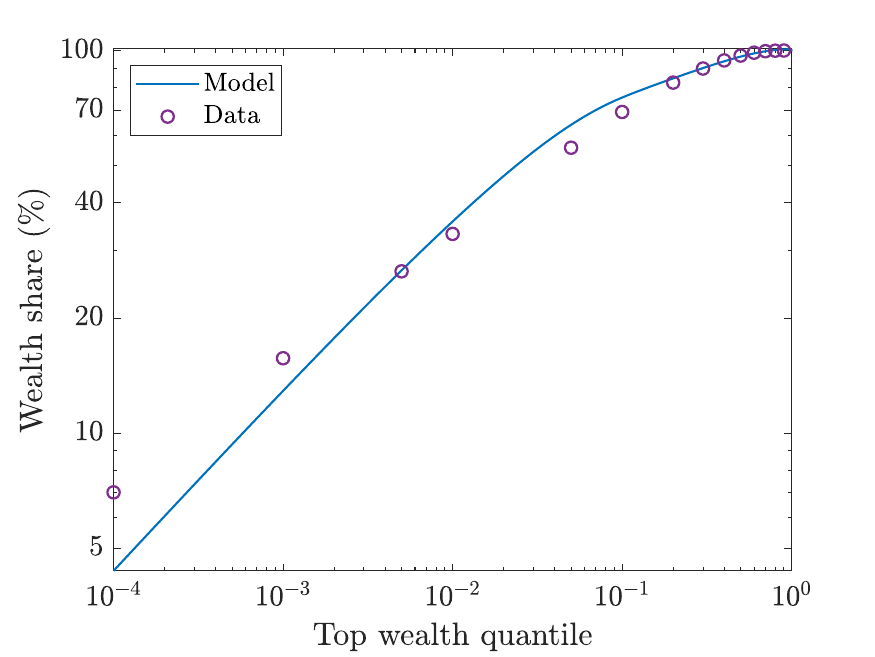}
			\caption{Wealth shares.}\label{fig:wealthShare}
		\end{subfigure}
		\caption{Stationary equilibrium distribution of financial wealth.}\label{fig:wealth}
	\end{figure}
	
	While the total wealth $S$ of an agent is guaranteed to be nonnegative, their financial wealth $W$ can in principle be negative: agents may, by borrowing at the risk-free rate, reduce their financial wealth to $R\ubar{b}=-22.9$. Only a small fraction of the agents in our model have negative financial wealth. It is not visible in Figure \ref{fig:tailProb}, but 2.5\% of agents have zero financial wealth (the newborn agents), and a further 3\% of agents have negative financial wealth.
	
	To assess the extent to which the stationary equilibrium distribution of financial wealth in our calibrated model resembles the distribution of wealth in the U.S.\ economy, we computed a range of wealth shares in our model and compared them to the corresponding values for the U.S.\ household wealth distribution in 2001.\footnote{We use the 2001 U.S.\ household wealth distribution as wealth shares are available for many quantiles: Table 5 of \citet{DaviesSandstromShorrocksWolff2011} reports the wealth shares of the top 1\%, 5\%, and all deciles; Table B1 on pp.~241--243 of the online appendix to \citet{saez-zucman2016} reports the wealth shares of the top 0.01\%, 0.1\%, 0.5\%, 1\%, 5\%, and 10\%. We use the numbers in \citet{saez-zucman2016} for the top 1\%, 5\% and 10\% wealth shares, which are very close to those in \citet{DaviesSandstromShorrocksWolff2011}.} Figure \ref{fig:wealthShare} and Table \ref{t:wealthShare} provide, respectively, graphical and numerical comparisons of the wealth shares in our model and in the data. Despite the simplicity of our model, the wealth shares it implies are similar to those in the data, except at the very highest levels of wealth. Our model slightly overestimates the top 1-10\% wealth shares and underestimates the top 0.1\% and 0.01\% wealth shares. We see that in our model, as in the data, the top 1\% hold about one third of wealth and the bottom 50\% hold little wealth.
	
	\begin{table}
		\small\centering
		\caption{Empirical and model-implied wealth shares (\%).}\label{t:wealthShare}
		\begin{tabular*}{\textwidth}{@{\extracolsep\fill}lrrlrr}
			\toprule
			\textbf{Wealth group} & \textbf{Data} & \textbf{Model} & \textbf{Wealth group} & \textbf{Data} & \textbf{Model} \\
			\midrule
			Top 0.01\%& 7.0& 4.4 & Bottom 90\% & 30.8 & 24.6\\
			Top 0.1\%& 15.7& 12.9 & Bottom 80\% & 17.4 & 15.3\\
			Top 0.5\%& 26.5& 26.6 & Bottom 70\% & 10.1 & 9.9\\
			Top 1\%& 33.2& 35.7 & Bottom 60\% & 5.6 & 6.2\\
			Top 5\%& 55.8& 64.0 & Bottom 50\% & 2.8 & 3.5\\
			Top 10\%& 69.2& 75.4 & Bottom 40\% & 1.1 & 1.6\\
			& & & Bottom 30\% & 0.2 & 0.3\\
			& & & Bottom 20\% & -0.1 & -0.5\\
			& & & Bottom 10\% & -0.2 & -0.9\\
			\bottomrule
		\end{tabular*}
	\end{table}
	
	The log-log plot of top wealth shares in Figure \ref{fig:wealthShare} shows a straight-line pattern for the top 1\%, both in our model and in the data. However, the data suggest a flatter slope than is implied by our model. If the wealth distribution has a Pareto upper tail with exponent $\zeta$, then it is straightforward to show that the slope of this log-log plot approaches $1-1/\zeta$ as the tail probability approaches zero; see, for instance, \citet[Eq.~4.3]{GouinBonenfantToda2023}. We therefore deduce from Figure \ref{fig:wealthShare} that our model-implied Pareto exponent of $1.93$ is too large to match the very top empirical wealth shares for the year 2001. To investigate whether the same is true for other years, in Figure \ref{fig:alphaHat} we plot estimates of the Pareto exponent in the U.S.\ for each year from 1913 to 2019 obtained by applying the minimum distance estimator introduced in \citet{TodaWang2021JAE} to the top 0.01\%, 0.1\%, 0.5\%, and 1\% wealth shares updated from \citet{PikettySaezZucman2018}.\footnote{Specifically, we obtain the wealth share series from \url{https://gabriel-zucman.eu/usdina} (Tables II: distributional series, tabs \texttt{TE2b} and \texttt{TE2c}) and apply the \textsc{Matlab} routine \texttt{ParetoCUMDE.m} available at \url{http://qed.econ.queensu.ca/jae/datasets/toda002}.} Estimates vary between 1 and 2, with an estimate of 1.47 for 2019. Thus our model-implied Pareto exponent of $1.93$ is on the high end of the empirical range. It roughly matches the estimated Pareto exponent for 1980, proceeded by sharply declining estimates throughout the 1980s.
	
	\begin{figure}
		\centering
		\includegraphics[width=\linewidth]{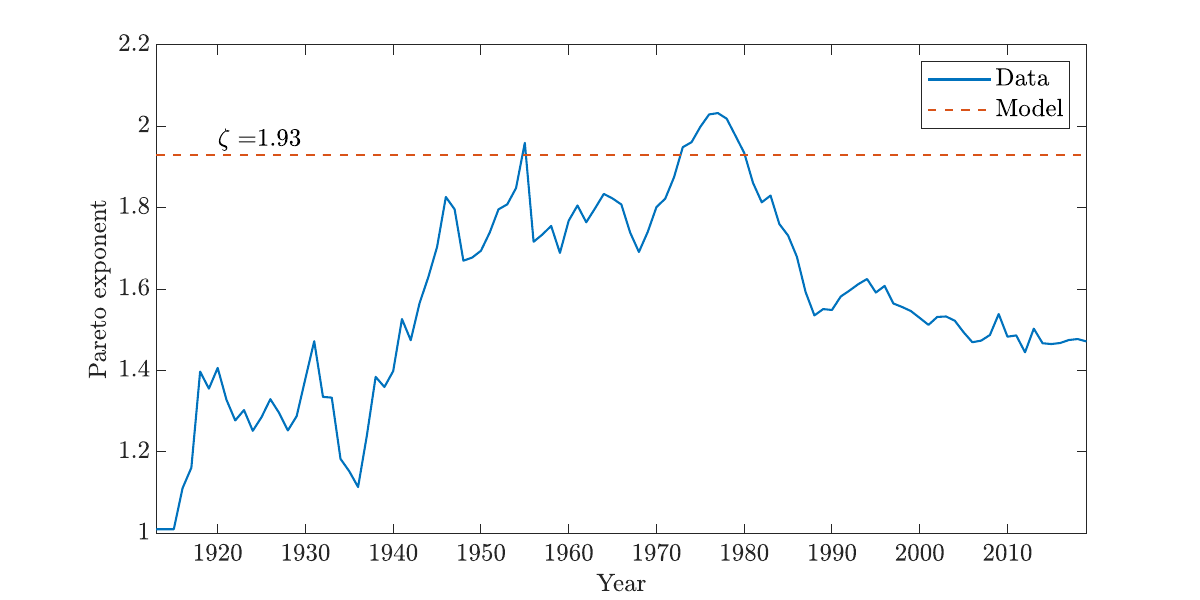}
		\caption{Estimated Pareto exponents for the U.S.\ household wealth distribution.}\label{fig:alphaHat}
	\end{figure}
	
	While our model-implied Pareto exponent is too high to match the very top wealth shares in current U.S.\ data, this ought not to draw into question the general conclusions of this article, which are that we should tax consumption rather than labor income, and probably also tax capital income. Indeed, with a smaller model-implied Pareto exponent, these conclusions would be strengthened. If a greater proportion of wealth were held by agents whose income is overwhelmingly derived from capital, then the taxation of labor income would be even less effective in drawing revenue from the top end of the wealth distribution. Though imperfect, our model succeeds in capturing the straight-line pattern of top wealth shares evident in Figure \ref{fig:wealthShare}, and correctly predicts that a significant fraction of wealth is held by agents well within the top 1\% of the wealth distribution. As shown in \cite{StachurskiToda2019}, economic models in which the sole source of heterogeneity is idiosyncratic variation in labor income do not typically generate a power law in the upper tail of the wealth distribution, and consequently fail to match the extreme degree of wealth concentration observed in data. It is the interaction of random multiplicative returns to capital with the random replacement of agents which, in our model, generates such a power law. Other models sharing this feature have been used to study wealth inequality in \cite{NireiAoki2016}, \cite{AokiNirei2017} and \cite{CaoLuo2017}. For a general survey of the literature on power laws in economics and finance, see \cite{Gabaix2009}.
	
	We mentioned in Section \ref{sec:equilibrium} that conditions ensuring the existence and uniqueness of prices solving the equilibrium conditions for the bond and labor markets are unavailable. In the numerical calibration under current discussion, and in variations to it discussed in the remainder of this article, we were able to identify equilibrium prices which are, as far as we can tell, unique. However, in unreported analysis, we were able to identify parameter configurations under which no equilibrium exists. Equilibrium nonexistence occurs when the probability of transitioning from entrepreneur to worker is too high. In this case entrepreneurs are unwilling to hold capital, fearing the large capital loss brought about by the depreciation of unutilized capital which will eventuate if they transition to the worker type. This issue is not problematic for our analysis as the specified transition rate is well below the level that would lead to equilibrium nonexistence, but serves to illustrate the impossibility of providing a general assertion of equilibrium existence in our model without imposing further technical conditions.
	
	\section{Optimal taxation}\label{sec:opttax}
	
	In this section we investigate the welfare implications of varying the labor income, capital income and consumption tax rates from the baseline rates of $0.25$, $0.4$ and $0$ used in the numerically calibrated model described in Section \ref{sec:numerical}. We only consider combinations of tax rates which generate the same aggregate tax revenue in stationary equilibrium as is achieved with the baseline rates. Proposition \ref{prop:tax} allows us to easily compute the aggregate tax revenue generated by any combination of tax rates, and thereby confine ourselves to combinations which generate the same aggregate tax revenue as the baseline rates.
	
	\subsection{Optimal taxation without a consumption tax}\label{sec:opttaxnocon}
	
	\begin{figure}
		\centering
		\begin{subfigure}{0.48\linewidth}
			\includegraphics[width=\linewidth]{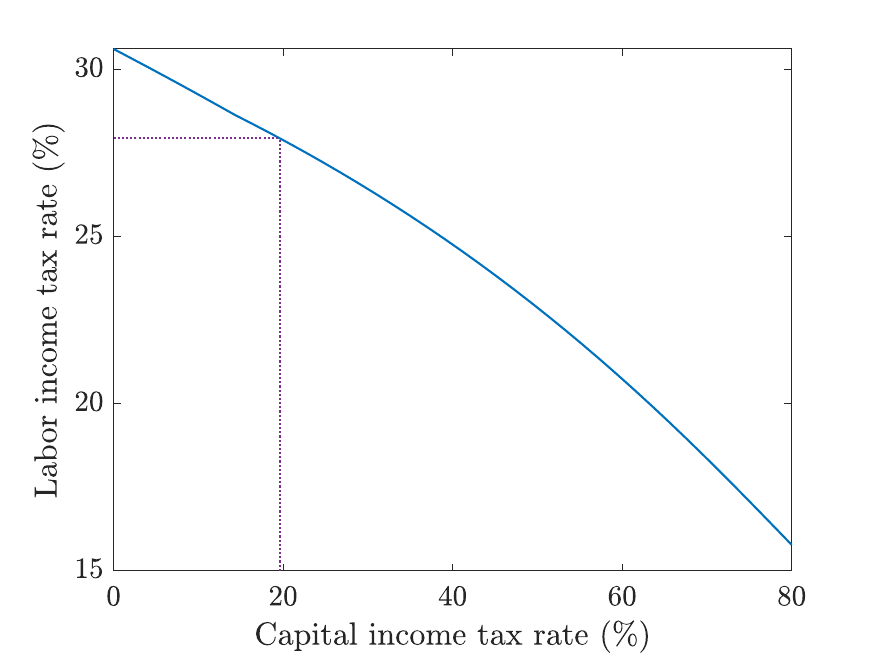}
			\caption{Revenue-preserving tax mixes.}\label{fig:tauK_tauL}
		\end{subfigure}
		\begin{subfigure}{0.48\linewidth}
			\includegraphics[width=\linewidth]{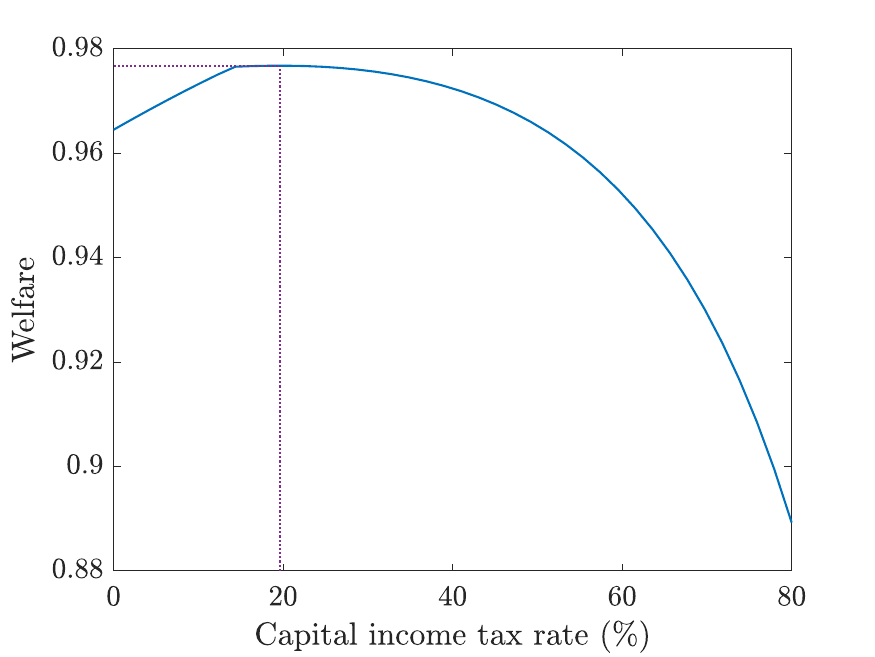}
			\caption{Welfare.}\label{fig:tauK_W}
		\end{subfigure}
		\begin{subfigure}{0.48\linewidth}
			\includegraphics[width=\linewidth]{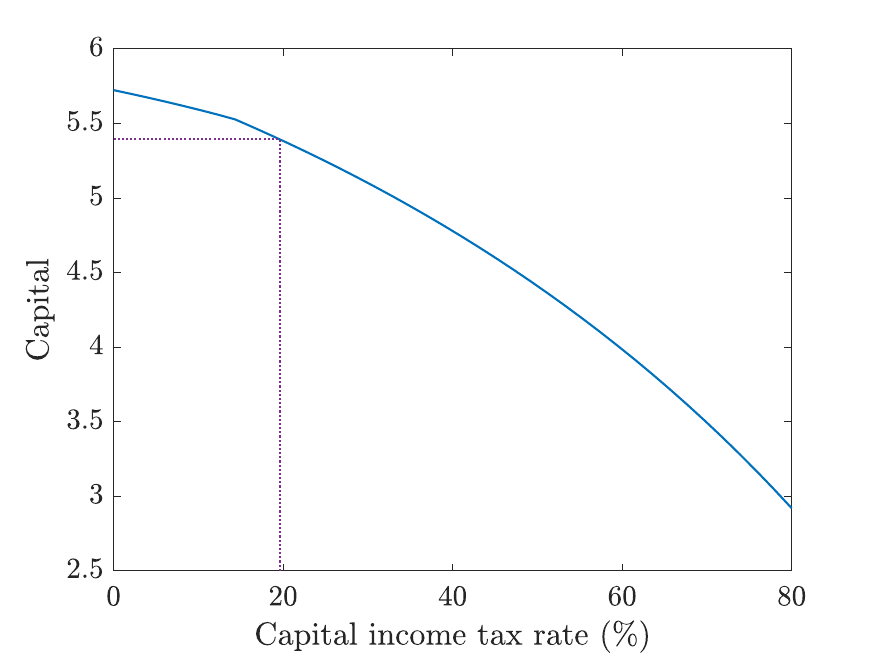}
			\caption{Aggregate capital.}\label{fig:tauK_AggW}
		\end{subfigure}
		\begin{subfigure}{0.48\linewidth}
			\includegraphics[width=\linewidth]{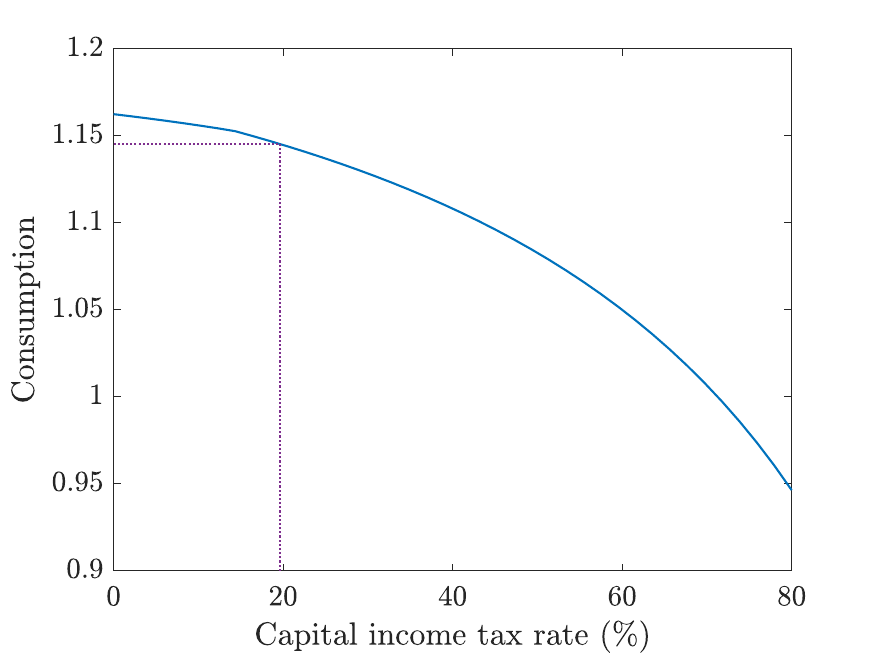}
			\caption{Aggregate consumption.}\label{fig:tauK_AggC}
		\end{subfigure}
		\begin{subfigure}{0.48\linewidth}
			\includegraphics[width=\linewidth]{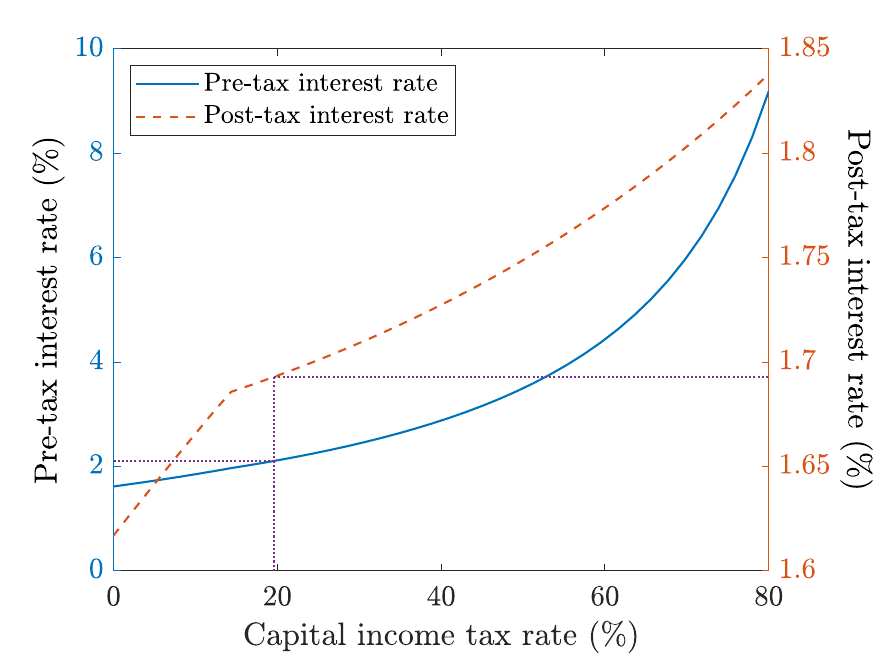}
			\caption{Interest rate.}\label{fig:tauK_R}
		\end{subfigure}
		\begin{subfigure}{0.48\linewidth}
			\includegraphics[width=\linewidth]{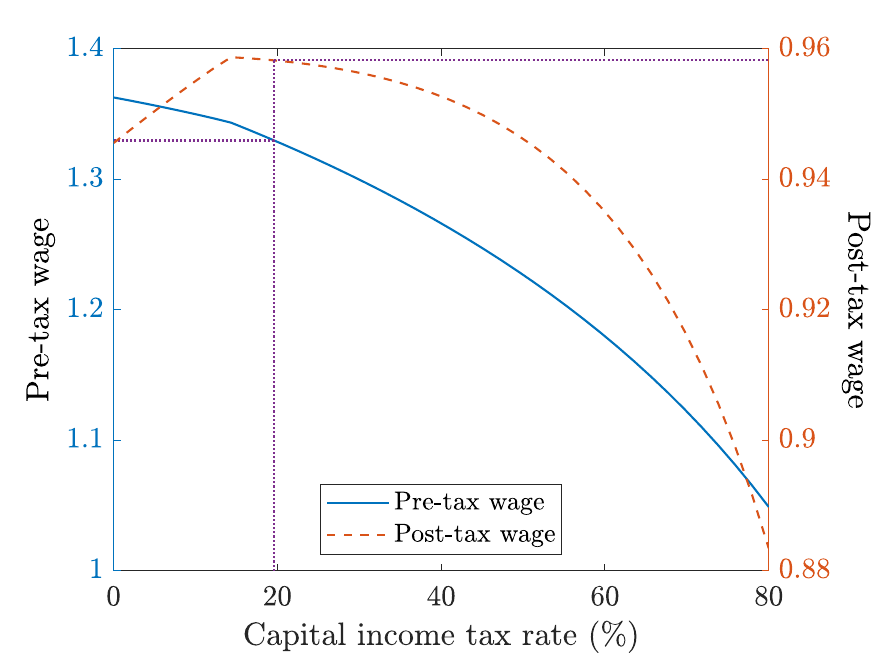}
			\caption{Wage.}\label{fig:tauK_omega}
		\end{subfigure}
		\caption{Effect of varying labor and capital income tax rates (no consumption tax).}\label{fig:tauK}
	\end{figure}
	
	Partly for expository purposes, and partly because it is interesting in its own right, we begin by exploring variations in the rates of labor and capital income taxation while holding the rate of consumption taxation equal to zero. Figure \ref{fig:tauK_tauL} displays the combinations of labor and capital income tax rates (with the latter constrained to fall between $0$ and $0.8$) which generate the same aggregate tax revenue as the baseline rates of $0.25$ and $0.4$. The revenue-preserving rate of labor income taxation declines monotonically with the rate of capital income taxation. Increasing the rate of labor income taxation to a little more than $0.3$ would allow the capital income tax to be completely eliminated, as has been advocated in parts of the literature \citep[see e.g.][]{Lucas1990,AtkesonChariKehoe1999}.
	
	In Figure \ref{fig:tauK_W} we plot the welfare (computed using Proposition \ref{prop:welfare}) that is obtained for a given rate of capital income taxation, with the rate of labor income taxation varying to preserve revenue as in Figure \ref{fig:tauK_tauL}. Welfare is hump-shaped, and maximized at a capital income tax rate of around $0.2$. The corresponding rate of labor income taxation is around $0.28$. Dotted lines in the four panels in Figure \ref{fig:tauK} identify these optimal rates of taxation.
	
	Two features of the welfare curve in Figure \ref{fig:tauK_W} deserve further comment. First, the welfare curve is \emph{very} flat around its maximum. It is visually indistinguishable from a horizontal line for capital income tax rates between around $0.14$ and $0.25$, and within 0.5\% of its maximum value for capital income tax rates between around $0.1$ and $0.41$. We calculated the welfare curve to a precision of ten decimal places and found it to be uniquely maximized with a capital income tax rate of around $0.2$, but for all practical purposes welfare achieves its maximum over a substantial range of capital income tax rates.
	
	The second notable feature of the welfare curve in Figure \ref{fig:tauK_W} is that it is kinked near the capital income tax rate of $0.14$. The natural borrowing constraint strictly binds on entrepreneurs to the right of the kink, is slack to the left of the kink, and is barely binding exactly at the kink. Equilibrium prices behave differently to the left and right of the kink, affecting welfare, as we discuss below. The Domar-Musgrave effect explains why the natural borrowing constraint binds on entrepreneurs only when the rate of capital income taxation is sufficiently high. With full offset provisions, a rate of capital income taxation in excess of $0.14$ mitigates the downside risk faced by entrepreneurs enough to induce them to fully leverage their capital investment. When the rate of capital income taxation falls below $0.14$, entrepreneurs reduce their leverage in an effort to limit their exposure to negative investment returns.
	
	Figures \ref{fig:tauK_AggW} and \ref{fig:tauK_AggC} show that aggregate capital and consumption are both decreasing in the capital income tax rate. It should not be surprising that the welfare maximizing rate of capital income taxation does not maximize aggregate consumption. The distributional impact of taxation policy, discussed further below, plays a critical role in the determination of welfare. It is a little difficult to see, but the aggregate capital and consumption curves are both kinked at the same location as the welfare curve. While aggregate capital continues to rise as the capital income tax rate falls below $0.14$, leverage is reduced, producing the kink in this curve.
	
	Figures \ref{fig:tauK_R} and \ref{fig:tauK_omega} show how the pre- and post-tax equilibrium interest rate and wage vary with the rate of capital income taxation. The curves plotted are kinked as in the panels above, although this is difficult to see in the pre-tax case. We see in Figure \ref{fig:tauK_R} that the equilibrium interest rate decreases as the rate of capital income taxation is reduced, but at a faster rate to the left of the kink. When the rate of capital income taxation is reduced below $0.14$, entrepreneurs are less willing to take on debt, so the interest rate falls more quickly to maintain equilibrium in the bond market. In Figure \ref{fig:tauK_omega} we see that the equilibrium pre-tax wage is strictly decreasing in the capital income tax rate, whereas the equilibrium post-tax wage is increasing to the left of the kink and decreasing to the right. The reason for the monotonicity of the pre-tax wage is clear: as the capital income tax rate is increased, aggregate capital falls, reducing the marginal productivity of labor and thus the demand for labor. For the post-tax wage, this effect is offset by the fact that the labor income tax rate falls as the capital income tax rate increases. The latter effect dominates to the left of the kink because aggregate capital is less sensitive to the rate of capital income taxation in this region. The lower post-tax wage to the left of the kink explains the sharp decline in welfare in Figure \ref{fig:tauK_W} which occurs when the capital income tax rate is reduced below $0.14$.
	
	\begin{figure}[b]
		\centering
		\begin{subfigure}{0.48\linewidth}
			\includegraphics[width=\linewidth]{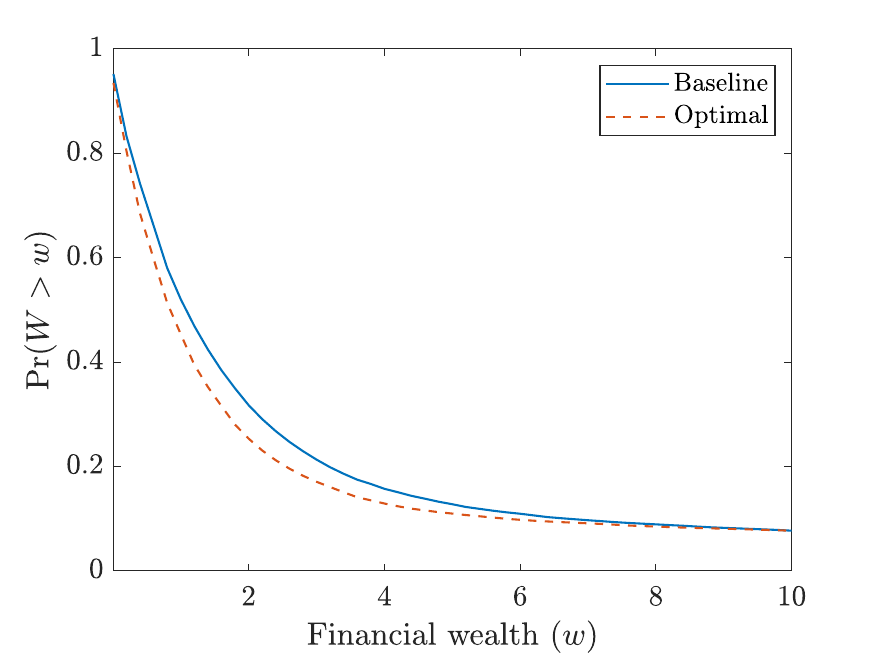}
			\caption{Exceedance probabilities for $w<10$.}\label{fig:W_opt_noC_body}
		\end{subfigure}
		\begin{subfigure}{0.48\linewidth}
			\includegraphics[width=\linewidth]{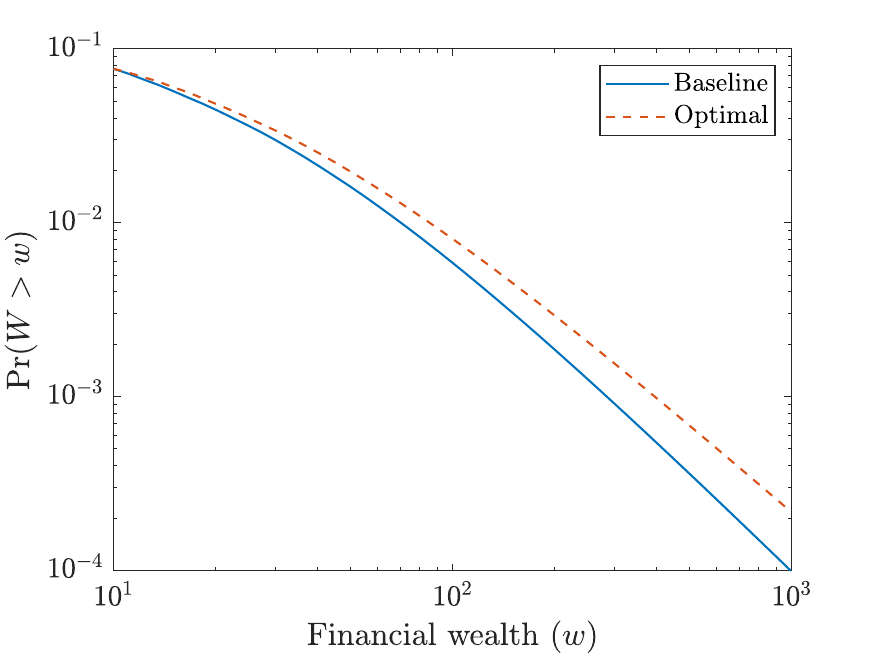}
			\caption{Exceedance probabilities for $w>10$.}\label{fig:W_opt_noC_tail}
		\end{subfigure}
		\caption{Baseline and optimal wealth distributions (consumption tax unavailable).}\label{fig:W_opt_noC}
	\end{figure}
	
	To provide further insight into the distributional changes brought about by changing the labor and capital income tax rates from the baseline rates of $0.25$ and $0.4$ to the optimal rates of $0.3$ and $0.2$, we display the corresponding distributions of financial wealth in Figure \ref{fig:W_opt_noC}, computed by Fourier inversion. The proportion of agents with financial wealth exceeding a threshold $w$ is plotted in Figure \ref{fig:W_opt_noC_body} for $w<10$ and in Figure \ref{fig:W_opt_noC_tail} for $w>10$, the latter plot in log-log scale. The curves for the two tax regimes cross near $w=10$, with around $8\%$ of agents having financial wealth greater than $10$ under either tax regime. Agents in the bottom $92\%$ of the wealth distribution are less wealthy under the optimal tax regime, while agents in the top $8\%$ of the wealth distribution are more wealthy. Shifting from the baseline tax rates to the optimal tax rates exacerbates wealth inequality and makes a small proportion of agents much wealthier at the expense of the majority.
	
	\begin{figure}[b]
		\centering
		\begin{subfigure}{0.48\linewidth}
			\includegraphics[width=\linewidth]{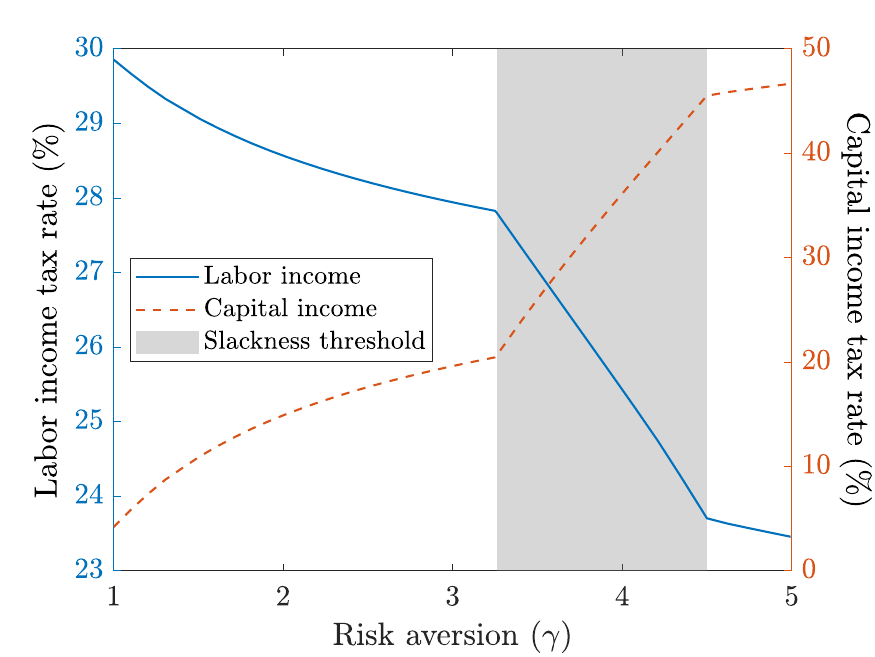}
			\caption{Optimal tax rates.}\label{fig:gamma_LK_tax}
		\end{subfigure}
		\begin{subfigure}{0.48\linewidth}
			\includegraphics[width=\linewidth]{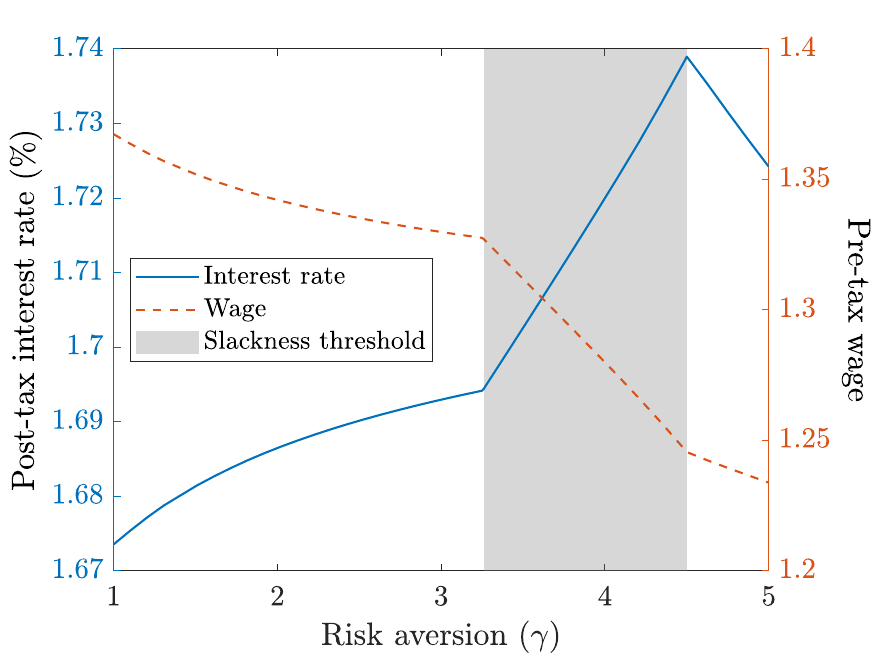}
			\caption{Equilibrium prices.}\label{fig:gamma_LK_prices}
		\end{subfigure}
		\caption{Effect of risk aversion on optimal taxes and equilibrium prices (no cons.\ tax).}\label{fig:riskaversion_LK}
	\end{figure}
	
	The risk aversion parameter $\gamma$ in our numerical calibration was set equal to $3$, following \cite{Angeletos2007}. In Figure \ref{fig:riskaversion_LK} we show how the optimal tax rates and corresponding equilibrium prices are affected as $\gamma$ varies between $1$ and $5$. A surprising phenomenon is apparent: there are now \emph{two} kinks in each of the curves plotted. The locations of the kinks divide the risk aversion parameter space into three regions. The natural borrowing constraint is strictly binding on entrepreneurs for $\gamma<3.26$ and is slack for $\gamma>4.48$. For $\gamma$ between these values the natural borrowing constraint is barely binding, meaning that it holds with equality but entrepreneurs would not increase their level of debt in its absence. We label the region of risk aversion levels where the natural borrowing constraint is barely binding the \emph{slackness threshold} in Figure \ref{fig:riskaversion_LK}. Within the slackness threshold, the optimal rate of capital income taxation rises rapidly with the level of risk aversion. Elsewhere it rises less rapidly. The fact that higher levels of risk aversion produce higher optimal rates of capital income taxation is a manifestation of the Domar-Musgrave effect: with greater risk aversion, agents more highly value the loss mitigation provided by capital income taxation. As $\gamma$ rises from $3.26$ to $4.48$, it is optimal to raise the rate of capital income taxation just enough to induce entrepreneurs to maintain full leverage of capital. The $0.4$ baseline rate of capital income taxation would be optimal if $\gamma=4.2$, a value well within the empirically plausible range of risk aversion levels.
	
	We present a similar sensitivity analysis for the volatility parameter $\sigma$ in Figure \ref{fig:volatility_LK}, fixing $\gamma=3$. The choice of $\sigma$ used in our baseline calibration is $0.247$, compared to $0.2$ in \cite{Angeletos2007}. As in Figure \ref{fig:riskaversion_LK}, we see that each of the curves plotted in Figure \ref{fig:volatility_LK} has two kinks. The natural borrowing constraint is strictly binding on entrepreneurs for $\sigma<0.25$, slack for $\sigma>0.27$, and barely binding between these values. As $\sigma$ rises from $0.25$ to $0.27$, the optimal capital income tax rate surges from $0.2$ to $0.43$. Therefore, raising our volatility parameter above the chosen value of $0.247$ by just a tiny amount leads the optimal rate of capital income taxation to more than double.
	
	\begin{figure}
		\centering
		\begin{subfigure}{0.48\linewidth}
			\includegraphics[width=\linewidth]{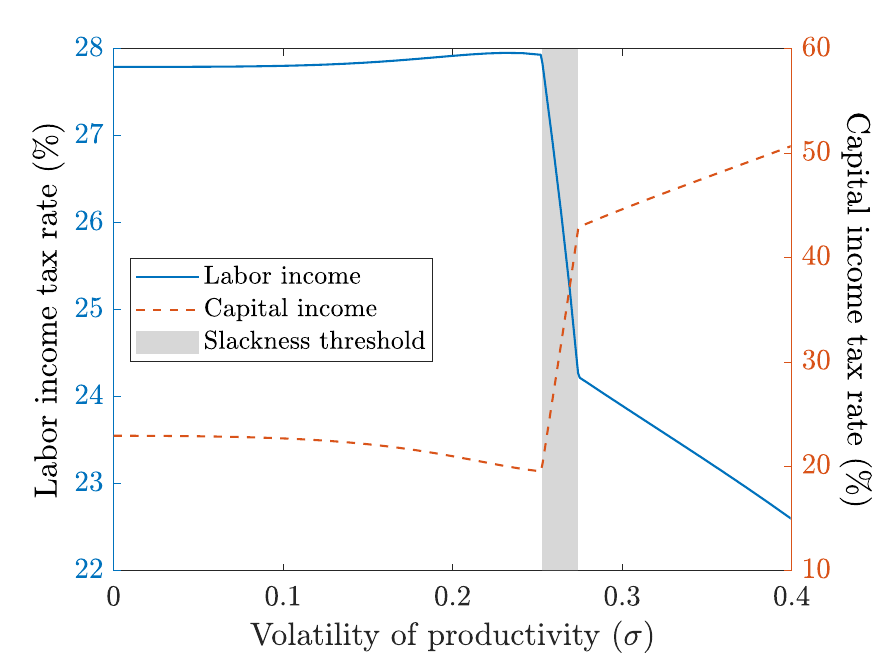}
			\caption{Optimal tax rates.}\label{fig:sigma_LK_tax}
		\end{subfigure}
		\begin{subfigure}{0.48\linewidth}
			\includegraphics[width=\linewidth]{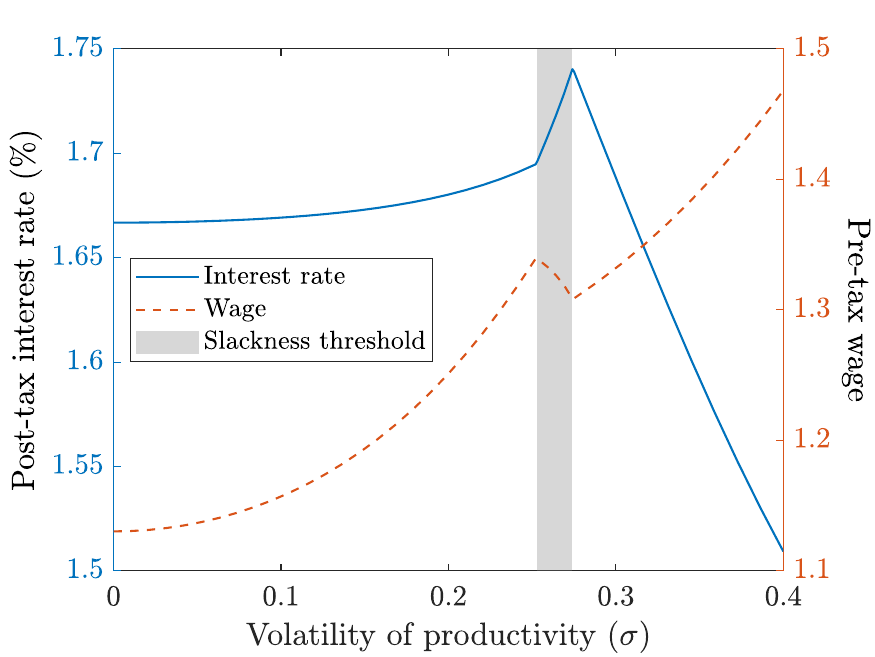}
			\caption{Equilibrium prices.}\label{fig:sigma_LK_prices}
		\end{subfigure}
		\caption{Effect of volatility on optimal taxes and equilibrium prices (no cons.\ tax).}\label{fig:volatility_LK}
	\end{figure}
	
	We summarize the results reported in this subsection as follows. In the absence of a consumption tax, our model does not provide a compelling justification for changing the existing balance between labor and capital income taxation. The optimal rates of taxation are highly sensitive to the assumed levels of risk aversion and return heterogeneity. Even ignoring this issue, changing from the baseline rates of taxation to the optimal rates would produce only a small increase in welfare of less than $0.5\%$ while substantially exacerbating wealth inequality. The interplay between the Domar-Musgrave effect and the natural borrowing constraint leads the optimal tax rates and equilibrium prices to behave in unexpected ways.

	\subsection{Optimal taxation with a consumption tax}\label{sec:opttaxcon}
	
	We now introduce a consumption tax to our model, and explore the welfare implications of varying the labor income, capital income and consumption tax rates while maintaining the aggregate tax revenue achieved in our baseline specification (without a consumption tax) described in Section \ref{sec:numerical}.
	
	The availability of a consumption tax introduces a third dimension to our analysis. In Figure \ref{fig:KL3d} we extend the two-dimensional plots provided in Figure \ref{fig:tauK} into this third dimension. Curves in two-dimensional space become surfaces in three-dimensional space. Figure \ref{fig:KLC} displays the combinations of labor income, capital income and consumption tax rates generating the target aggregate tax revenue. The lower edge of the surface plotted, where the consumption tax rate is zero, corresponds precisely to the pairs of labor and capital income tax rates plotted in Figure \ref{fig:tauK_tauL}. When the labor and/or capital income tax rates are reduced below the levels on this edge, the consumption tax rate rises so as to preserve aggregate tax revenue.
	
	\begin{figure}[!p]
		\centering
		\begin{subfigure}{0.45\linewidth}
			\includegraphics[width=\linewidth]{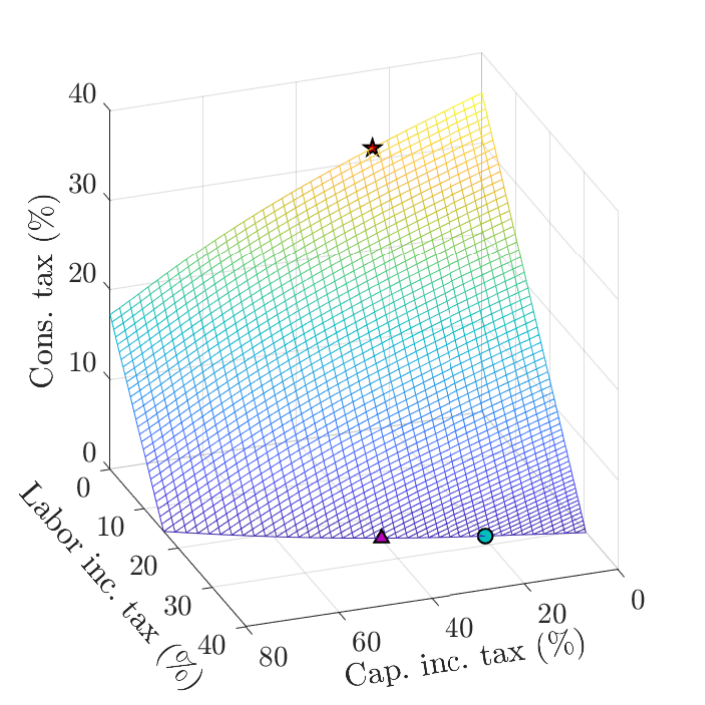}
			\caption{Revenue-preserving tax mixes.}\label{fig:KLC}
		\end{subfigure}
		\hspace{1cm}
		\begin{subfigure}{0.45\linewidth}
			\includegraphics[width=\linewidth]{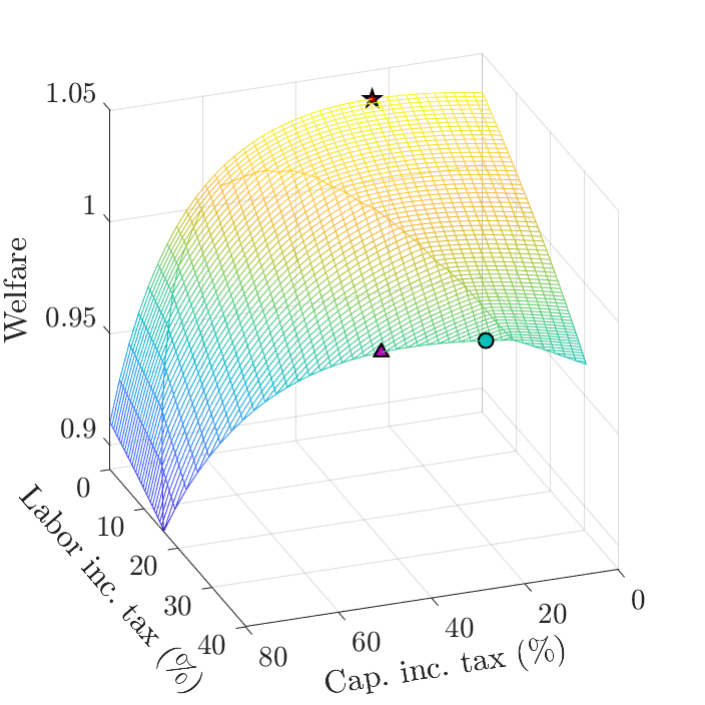}
			\caption{Welfare.}\label{fig:KLWnew}
		\end{subfigure}
		\begin{subfigure}{0.45\linewidth}
			\includegraphics[width=\linewidth]{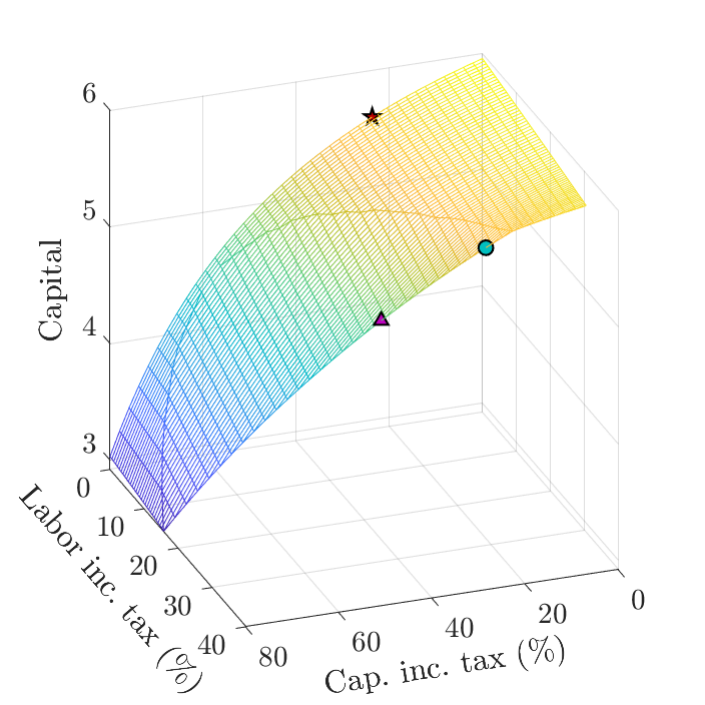}
			\caption{Aggregate capital.}\label{fig:KLAggK}
		\end{subfigure}
		\hspace{1cm}
		\begin{subfigure}{0.45\linewidth}
			\includegraphics[width=\linewidth]{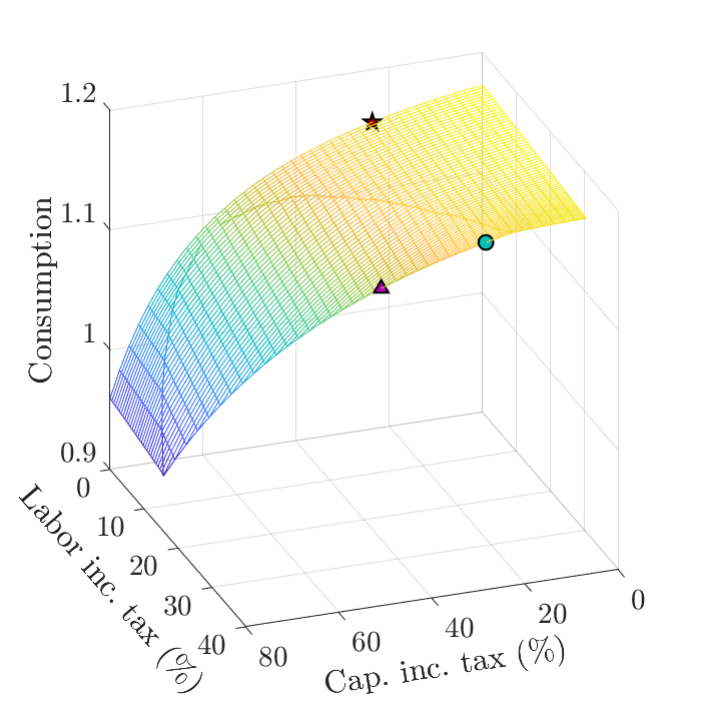}
			\caption{Aggregate consumption.}\label{fig:KLAggC}
		\end{subfigure}
		\begin{subfigure}{0.45\linewidth}
			\includegraphics[width=\linewidth]{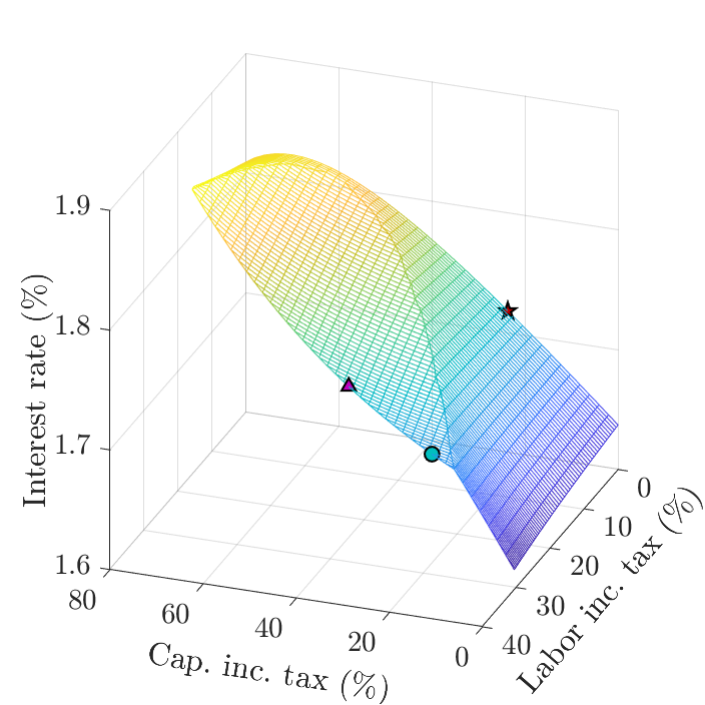}
			\caption{Post-tax interest rate.}\label{fig:KLR}
		\end{subfigure}
		\hspace{1cm}
		\begin{subfigure}{0.45\linewidth}
			\includegraphics[width=\linewidth]{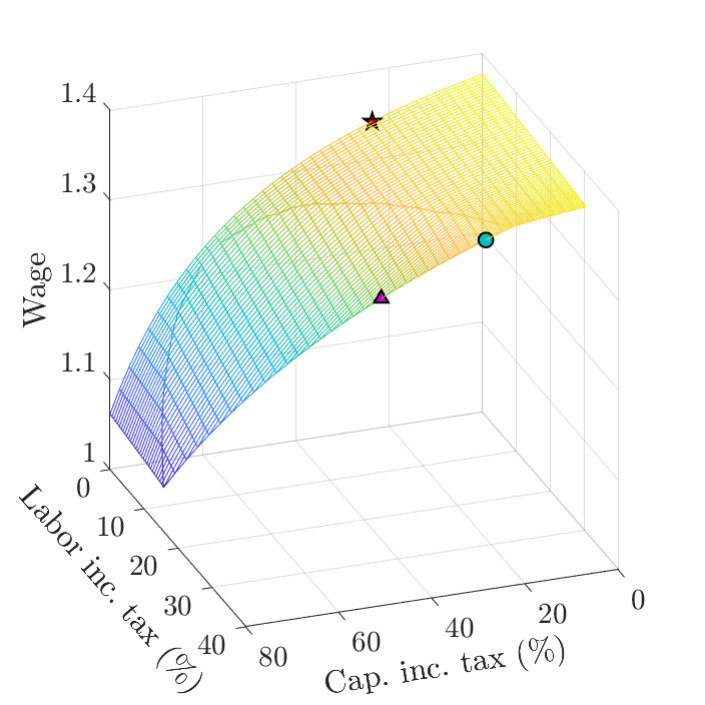}
			\caption{Pre-tax wage.}\label{fig:KLomega}
		\end{subfigure}
		\caption{Effect of varying labor income, capital income and consumption tax rates.}
		\label{fig:KL3d}
	\end{figure}
	
	Figure \ref{fig:KLWnew} displays the welfare that is obtained for a given pair of labor and capital income tax rates, with the consumption tax rate varying as in Figure \ref{fig:KLC} so as to preserve revenue. The star-shaped marker indicates the global maximum of welfare. The circle-shaped marker indicates the welfare maximum with no consumption tax. The triangle-shaped marker indicates the welfare level obtained with the baseline tax rates. The three markers appear in all six panels of Figure \ref{fig:KL3d}. We see that the labor income tax rate is zero at the global welfare optimum. The optimal rates of capital income taxation and consumption taxation are $0.24$ and $0.31$ respectively. Notably, it is not optimal to eliminate or nearly eliminate all capital income taxation in favor of consumption taxation. In this respect our findings differ from those reported in prior studies (using different models) including \cite{Imrohoroglu1998} and \cite{Coleman2000}, as discussed in Section \ref{sec:intro}. On the other hand, consistent with these prior studies, we find that the welfare achieved by only taxing consumption is substantially higher than can be achieved without taxing consumption, and not much less than the global welfare optimum. The global welfare optimum is 0.5\% higher than can be achieved by only taxing consumption, 6.2\% higher than can be achieved by only taxing labor and capital income, and 6.6\% higher than is achieved using the baseline rates of labor and capital income taxation with no consumption tax.
	
	The kink in the curve plotted in Figure \ref{fig:tauK_W} manifests as a curved ridge in the surface plotted in Figure \ref{fig:KLWnew}. A similar ridge is visible in Figures \ref{fig:KLAggK}--\ref{fig:KLomega}. Within the region enclosed by the ridge (where the circle- and triangle-shaped markers are located), the natural borrowing constraint is strictly binding on entrepreneurs. On the other side of the ridge (where the star-shaped marker is located), it is slack. Exactly on the ridge, it is barely binding.
	
	Figures \ref{fig:KLAggK}--\ref{fig:KLomega} show that aggregate capital and consumption and the pre-tax equilibrium wage are decreasing in the capital (labor) income tax rate if the labor (capital) income tax rate is held constant and the consumption tax rate varied to preserve revenue. The shape of the post-tax equilibrium interest rate surface is more complicated. Aggregate consumption at the global welfare optimum is 1\% higher than at the optimum without a consumption tax, and 4.3\% higher than at the baseline rates of labor and capital income taxation with no consumption tax. Aggregate capital at the global welfare optimum is 4.1\% higher than at the optimum without a consumption tax, and 17.1\% higher than at the baseline rates.
	
	It is interesting to observe in Figure \ref{fig:KLAggK} that substituting labor income taxation for consumption taxation while holding the capital income tax rate constant has a slight negative effect on aggregate capital. In our model, consumption taxation is nondistortionary (i.e., does not affect the investment decisions of agents) because all agents spend a fixed proportion of their total wealth on consumption and the consumption tax (Proposition \ref{prop:optrule}) and invest the remainder. If labor income taxation were also nondistortionary then we might expect the substitution of labor income taxation for consumption taxation to have no effect on aggregate capital. In fact, labor income taxation is mildly distortionary in our model, despite labor being inelastically supplied. The mechanism through which distortion takes place is the increase in human wealth (present value of future post-tax wages) brought about by a reduction in the labor income tax rate. While human wealth comprises only a tiny fraction of the total wealth of the wealthiest entrepreneurs, it can comprise a large fraction of the total wealth of less wealthy entrepreneurs. When the labor income tax rate falls, part of the associated increase in the total wealth of entrepreneurs is invested in capital, leading to an increase in aggregate capital.
	
	\begin{figure}[b]
		\centering
		\begin{subfigure}{0.48\linewidth}
			\includegraphics[width=\linewidth]{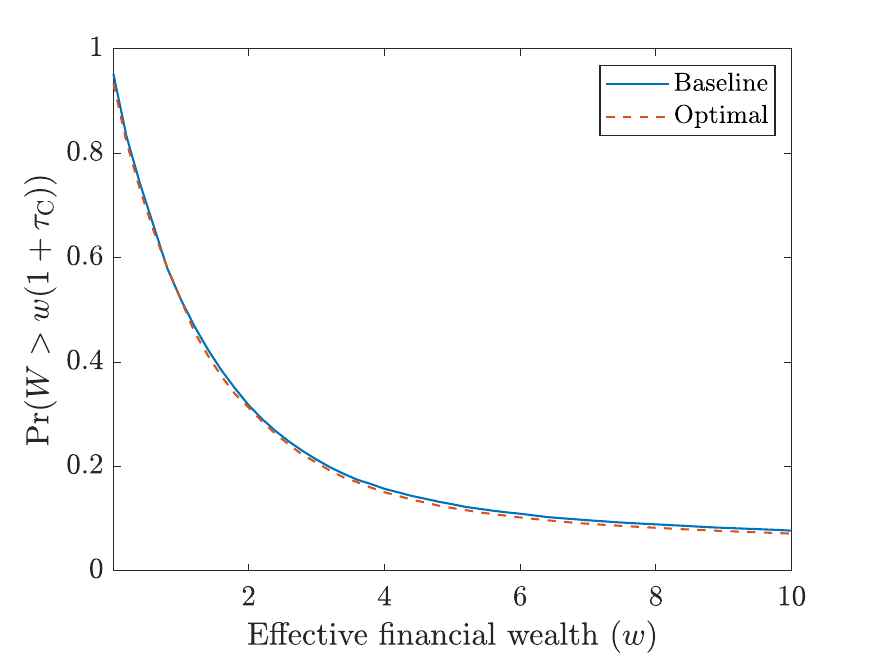}
			\caption{Exceedance probabilities for $w<10$.}\label{fig:W_opt_body}
		\end{subfigure}
		\begin{subfigure}{0.48\linewidth}
			\includegraphics[width=\linewidth]{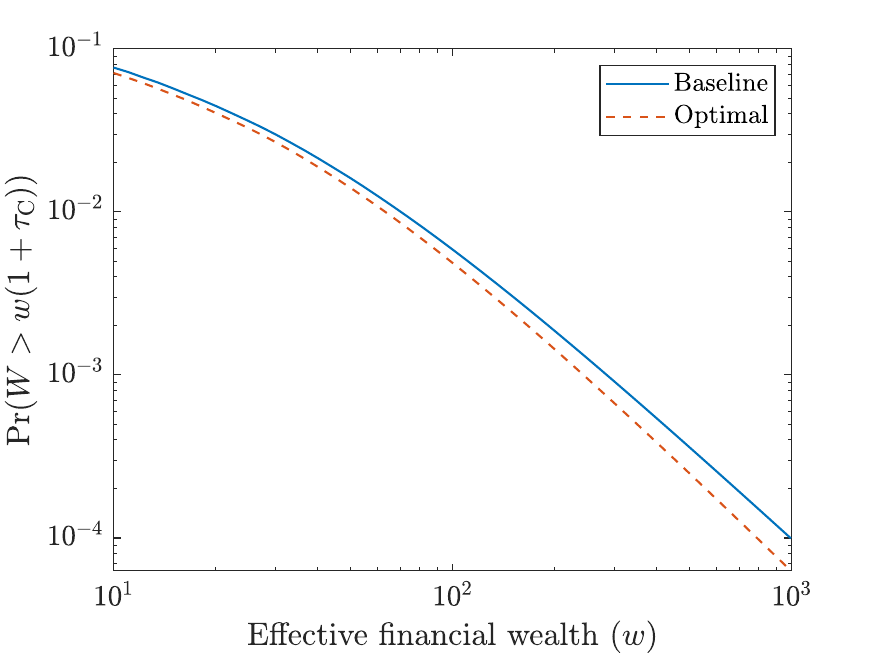}
			\caption{Exceedance probabilities for $w>10$.}\label{fig:W_opt_tail}
		\end{subfigure}
		\caption{Baseline and optimal wealth distributions (consumption tax available).}\label{fig:W_optC}
	\end{figure}
	
	In Figure \ref{fig:W_optC} we display the distributions of effective financial wealth arising under the optimal tax regime (which has no labor income tax) and under the baseline tax rates ($0.25$ for labor income and $0.4$ for capital income, with no consumption tax). We define effective financial wealth to be the value of financial wealth measured in terms of the post-tax price of one unit of consumption. The distribution displayed for the baseline rates is the same as the one in Figure \ref{fig:W_opt_noC}. Figure \ref{fig:W_optC} reveals that the distribution of effective financial wealth under the optimal tax regime differs very little from the distribution under the baseline tax rates for wealth levels less than 10 (about 92\% of agents). The top 8\% of agents are less wealthy under the optimal tax regime, increasingly so toward the top end of the distribution. It appears that the distribution of effective financial wealth under the optimal tax regime is first-order stochastically dominated by the distribution under the baseline tax rates. Welfare is, nevertheless, substantially higher under the optimal tax regime. The bulk of agents choose to maintain about the same level of effective financial wealth under the optimal tax regime as they do under the baseline tax rates, but consume more under the optimal tax regime.
	
	\begin{figure}[!b]
		\centering
		\begin{subfigure}{0.48\linewidth}
			\includegraphics[width=\linewidth]{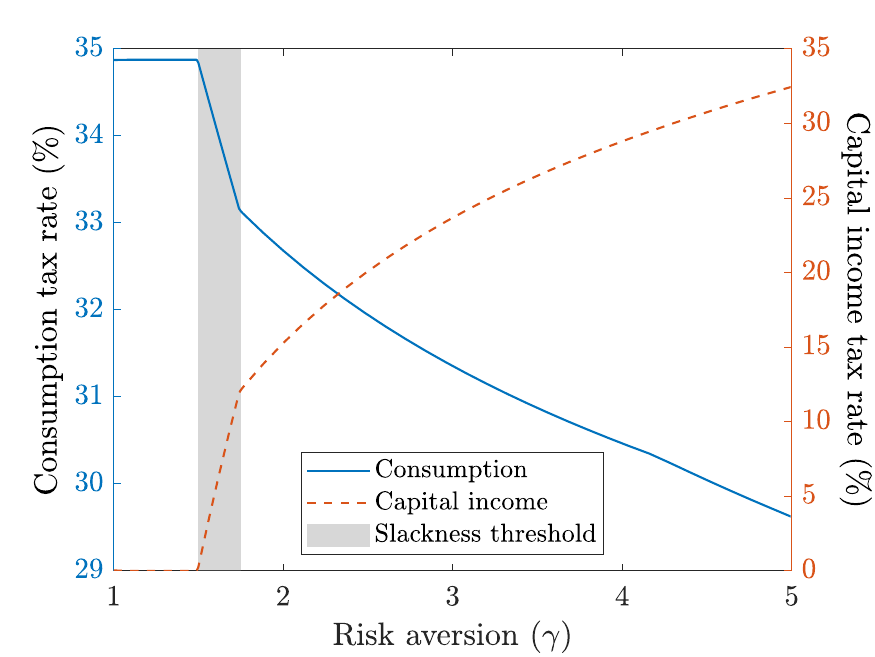}
			\caption{Optimal tax rates.}\label{fig:gamma_tax_CK}
		\end{subfigure}
		\begin{subfigure}{0.48\linewidth}
			\includegraphics[width=\linewidth]{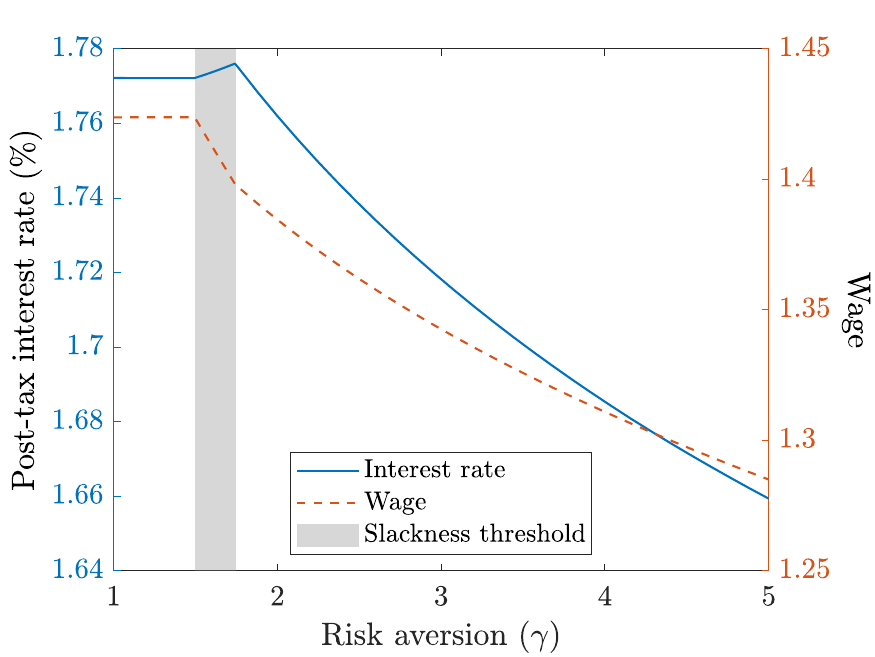}
			\caption{Equilibrium prices.}\label{fig:gamma_price_CK}
		\end{subfigure}
		\caption{Effect of risk aversion on optimal taxes and equilibrium prices.}\label{fig:riskaversion}
	\end{figure}
	
	Figures \ref{fig:riskaversion} and \ref{fig:volatility} show how the optimal tax rates and corresponding equilibrium prices are affected when we vary the risk aversion parameter $\gamma$ or the volatility of productivity parameter $\sigma$. As in Figures \ref{fig:riskaversion_LK} and \ref{fig:volatility_LK}, the parameter spaces for $\gamma$ and $\sigma$ are divided into three regions depending on whether the natural borrowing constraint binds on entrepreneurs. As either parameter increases from a low level, the borrowing constraint is first strictly binding, then barely binding, then slack. The optimal labor income tax rate is always zero, so we display only the optimal capital income and consumption tax rates. The optimal capital income tax rate is zero when the natural borrowing constraint is strictly binding. This occurs for $\gamma<1.5$ (with $\sigma=0.247$) and for $\sigma<0.2$ (with $\gamma=3$). In these cases it is optimal to generate all revenue through the taxation of consumption. The optimal rate of capital income taxation rises steeply as we move through the region in which the borrowing constraint is barely binding. It increases from $0$ to $0.12$ as $\gamma$ increases from $1.5$ to $1.75$ (with $\sigma=0.247$) or from $0$ to $0.18$ as $\sigma$ increases from $0.2$ to $0.215$ (with $\gamma=3$). Further increases to $\gamma$ or $\sigma$ produce less rapid increases in the optimal capital income tax rate. Overall, the optimal rate of capital income taxation varies substantially over a plausible range of parameter values; it may be zero, or perhaps as high as $0.4$. The corresponding range of consumption tax rates is between about $0.27$ and $0.35$.

	\begin{figure}[t!]
		\centering
		\begin{subfigure}{0.48\linewidth}
			\includegraphics[width=\linewidth]{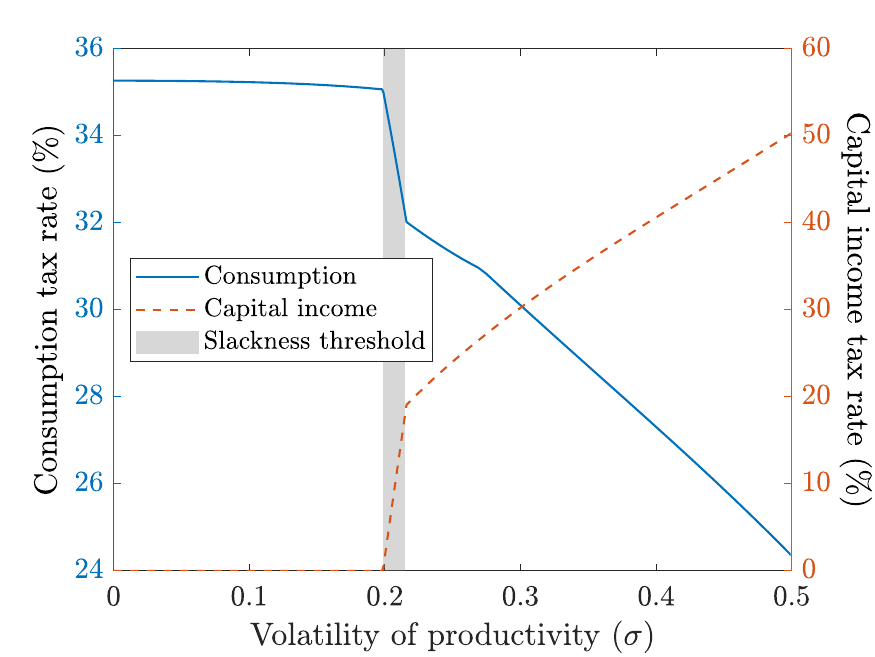}
			\caption{Optimal tax rates.}\label{fig:sigma_tax_CK}
		\end{subfigure}
		\begin{subfigure}{0.48\linewidth}
			\includegraphics[width=\linewidth]{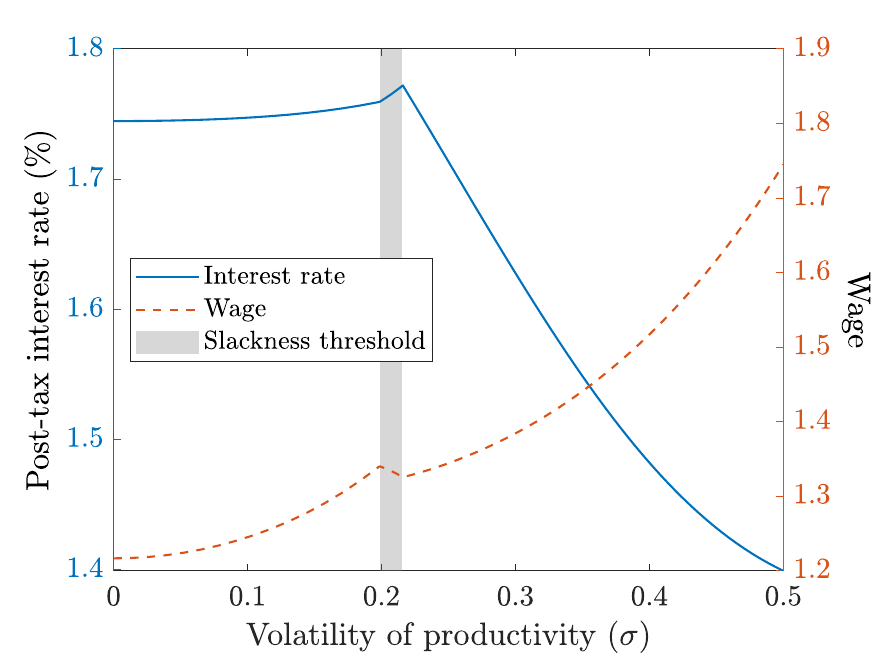}
			\caption{Equilibrium prices.}\label{fig:sigma_price_CK}
		\end{subfigure}
		\caption{Effect of volatility on optimal taxes and equilibrium prices.}\label{fig:volatility}
	\end{figure}

	\section{Transition to the consumption tax equilibrium}\label{sec:transition}
	
	We now turn to an analysis of the path along which our economy transitions to its new stationary equilibrium when the benchmark tax regime with a tax rate of $0.25$ on labor income and $0.4$ on capital income is replaced with the optimal tax regime with a tax rate of $0.24$ on capital income and $0.31$ on consumption. We know from our previous discussion that aggregate capital must rise by 17.1\% and aggregate consumption by 4.3\% along the transition path. However, the increase in aggregate capital cannot come from nowhere; there must necessarily be a period of reduced aggregate consumption during which resources are diverted toward capital accumulation.
	
	Determining the transition path of the economy is more computationally challenging than the analysis undertaken in Sections \ref{sec:numerical} and \ref{sec:opttax}. In stationary equilibrium, there are two equilibrium prices (the interest rate and wage) which are determined by solving two equations (bond and labor market clearing) simultaneously. To identify the path along which an economy transitions to a stationary equilibrium, we need to find a path of price pairs which clears the bond and labor markets in every period. The bond and labor market conditions cannot be solved separately for each period because agents base their decisions not only upon the current prices, but also all future prices, which they correctly anticipate. We describe the procedure used to compute equilibrium price paths at the end of this section.
	
	\begin{figure}
		\centering
		\begin{subfigure}{0.48\linewidth}
			\includegraphics[width=\linewidth]{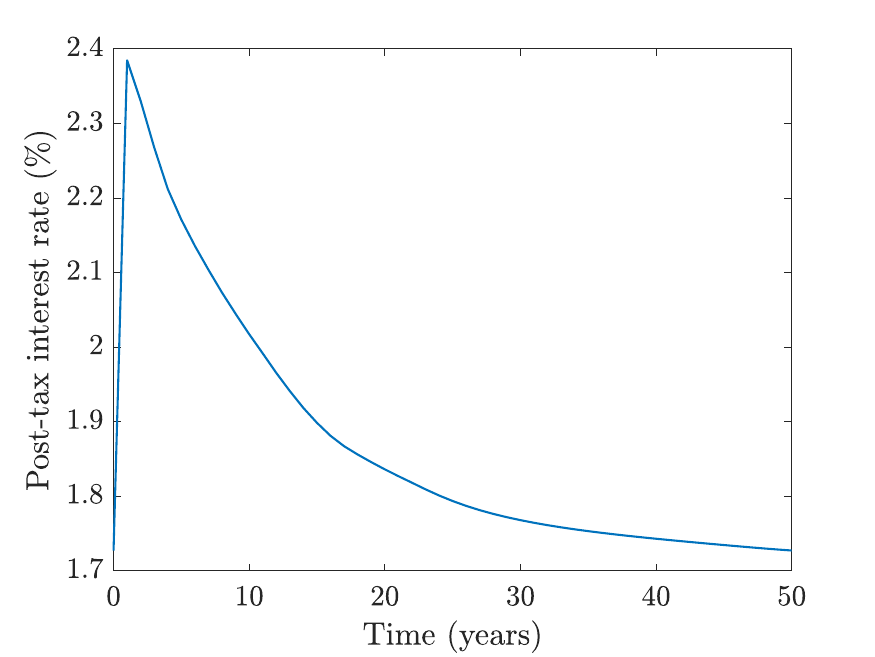}
			\caption{Interest rate.}\label{fig:trans_R}
		\end{subfigure}
		\begin{subfigure}{0.48\linewidth}
			\includegraphics[width=\linewidth]{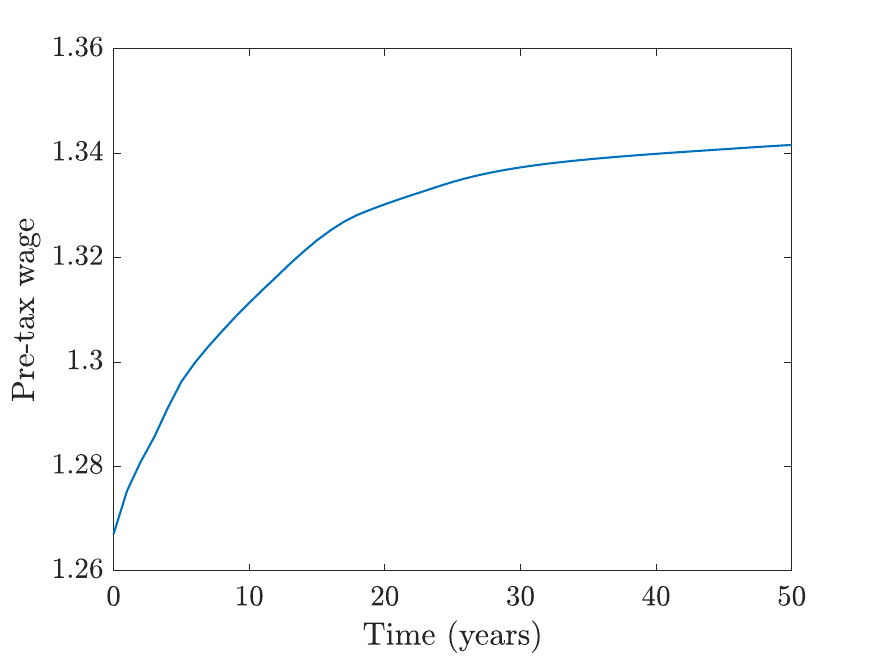}
			\caption{Wage.}\label{fig:trans_omega}
		\end{subfigure}
		\begin{subfigure}{0.48\linewidth}
			\includegraphics[width=\linewidth]{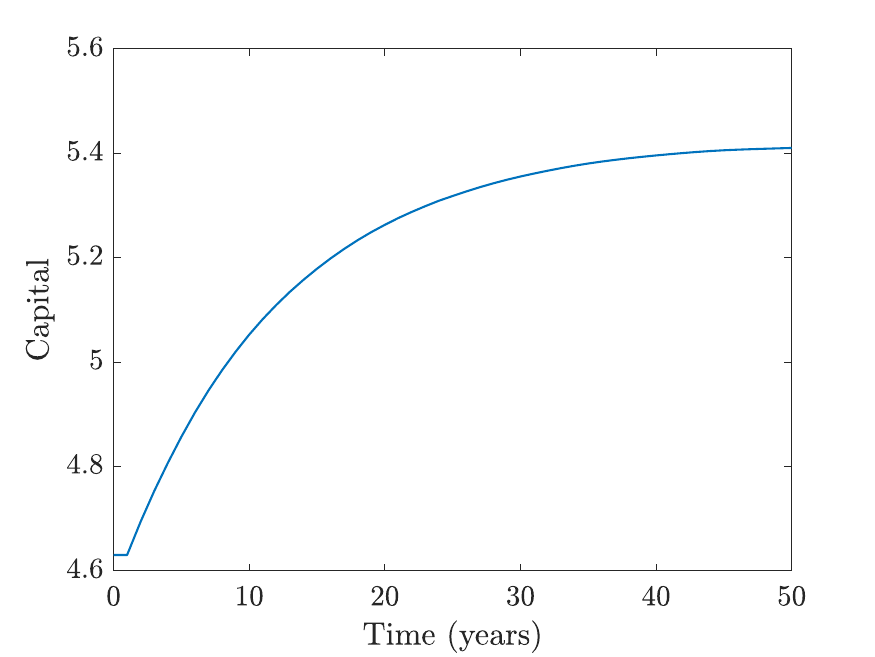}
			\caption{Aggregate capital.}\label{fig:trans_K}
		\end{subfigure}
		\begin{subfigure}{0.48\linewidth}
			\includegraphics[width=\linewidth]{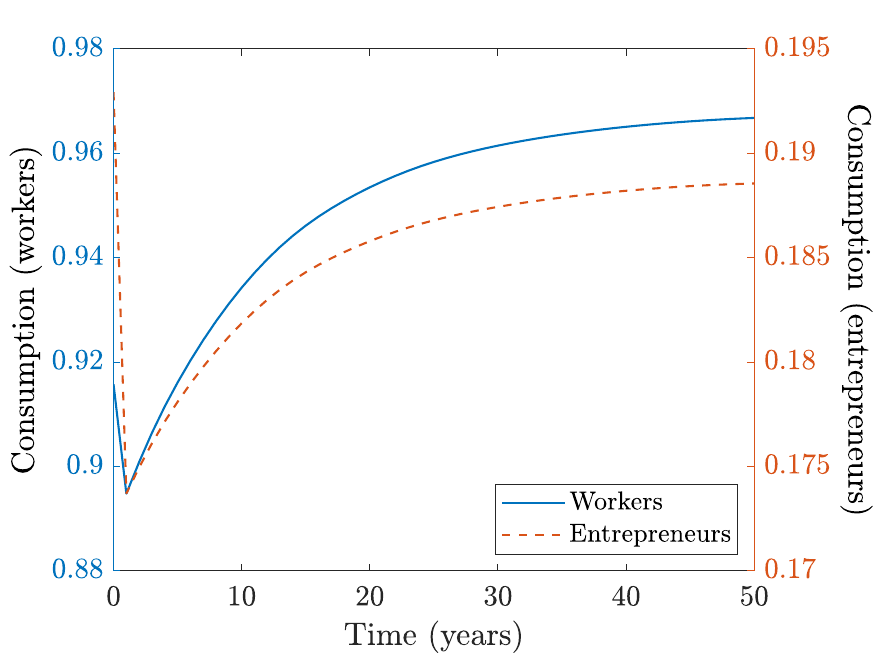}
			\caption{Aggregate consumption.}\label{fig:trans_C}
		\end{subfigure}
		\begin{subfigure}{0.48\linewidth}
			\includegraphics[width=\linewidth]{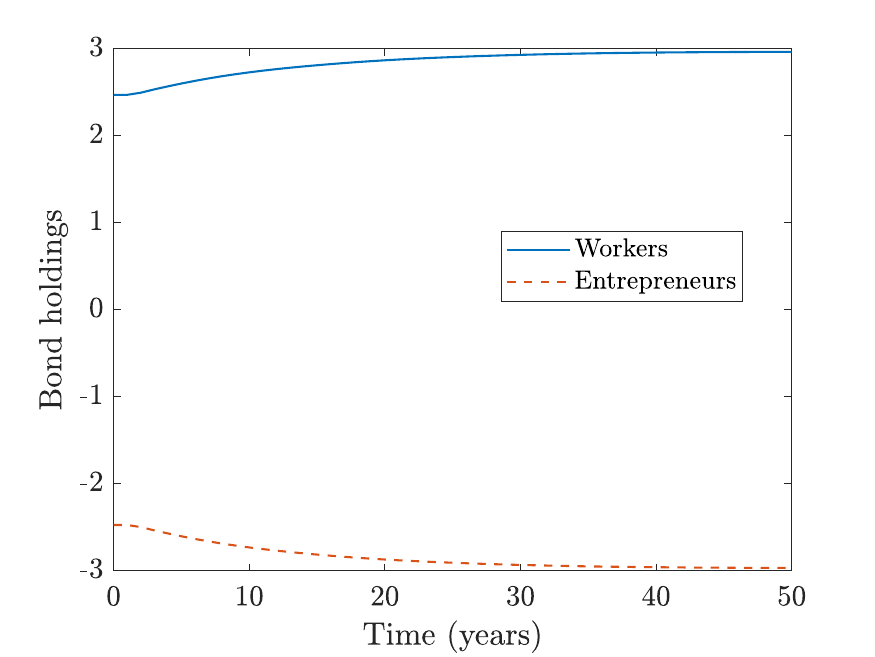}
			\caption{Aggregate bond holdings.}\label{fig:trans_B}
		\end{subfigure}
		\begin{subfigure}{0.48\linewidth}
			\includegraphics[width=\linewidth]{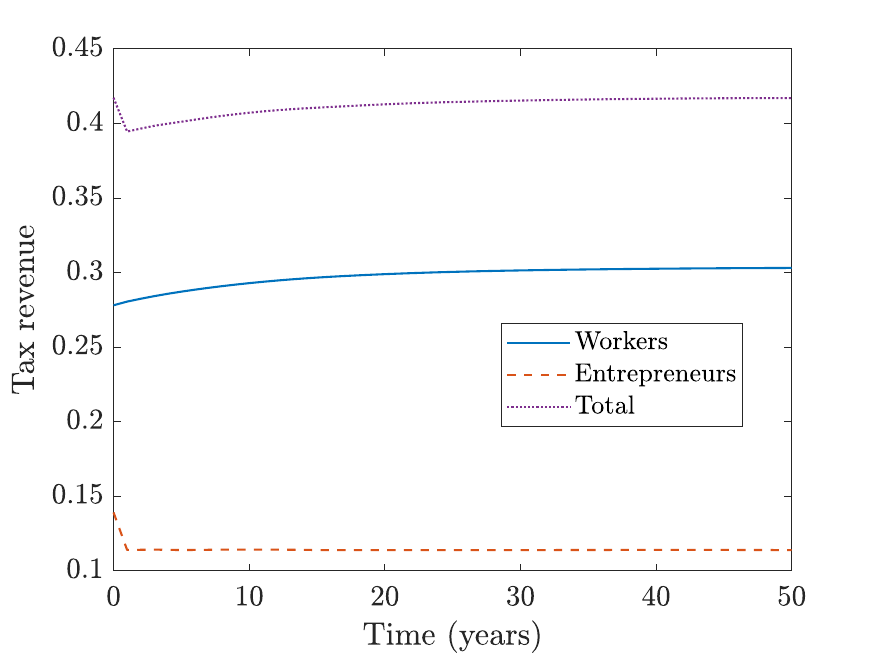}
			\caption{Aggregate tax revenue.}\label{fig:trans_revenue}
		\end{subfigure}
		\caption{Transition to stationary equilibrium with optimal taxes.}\label{fig:opt_trans}
	\end{figure}
	
	Figure \ref{fig:opt_trans} displays the computed transition to the new stationary equilibrium. We see that the transition is largely complete within $50$ years. The temporary period of depressed consumption is apparent in Figure \ref{fig:trans_C}. The change in tax rates in year one causes the aggregate consumption of workers to immediately drop by $2.3\%$, and that of entrepreneurs by $10\%$. Aggregate consumption subsequently rises. It surpasses its initial level within $6$ years for workers and within $10$ years for all agents, but never fully recovers for entrepreneurs. A temporary increase in the equilibrium interest rate accompanies the period of depressed consumption, as shown in Figure \ref{fig:trans_R}. This is the mechanism by which workers are encouraged to reduce their consumption and save more, these resources being redirected to the accumulation of capital. Figure \ref{fig:trans_B} shows the increase in the aggregate bond holdings of workers following the change in tax rates, and the accompanying increase in the aggregate debt of entrepreneurs. 
	
	The transition path of tax revenue is displayed in Figure \ref{fig:trans_revenue}. It is equal in years zero and $100$ by construction, and modestly reduced in the years between. The change in tax rates in year one causes tax revenue to immediately drop by $5.4\%$. This is entirely due to the reduction in the rate of capital income taxation. In subsequent years, consumption tax revenue drawn from workers and entrepreneurs rises, with total tax revenue nearly recovering to its initial level within 25 years. Interestingly, the total tax paid by entrepreneurs is exactly flat from year one onward. During this time, the consumption tax paid by entrepreneurs gradually rises, but the capital income tax paid by entrepreneurs falls by the same amount. The reduction in capital income tax revenue occurs due to the rising cost of labor (Figure \ref{fig:trans_omega}), and despite the gradual increase in aggregate capital (Figure \ref{fig:trans_K}).
	
	Would a majority of the agents in our model vote to change from the baseline tax regime to the optimal tax regime if offered such a choice at the beginning of year one? To answer this question, we computed (by Fourier inversion, using Proposition \ref{prop:stationary}) the proportion of agents for which the year one value function $V^\ast_{J_1}(S_1)$ is increased by changing the tax regime in year one. We find that switching to the optimal tax regime raises the value function for $86\%$ of all agents. A vote on whether to change the tax regime would therefore pass with overwhelming popular support. This support, however, is concentrated among workers, who constitute a large majority ($88.5\%$) of agents. The fraction of workers supporting the change in tax regime is $93\%$, compared to only $26\%$ of entrepreneurs.
	
	The procedure we used to compute the transition paths plotted in Figure \ref{fig:opt_trans} is, by necessity, somewhat complicated. In principle, the transition to stationary equilibrium need not be achieved in finite time, so there are infinitely many equilibrium prices to compute along the transition path. To reduce the problem from an infinite dimensional one to a finite dimensional one, we based our computations on the (false, but approximately correct) assumption that the economy completes its transition to stationary equilibrium within 100 years. Under this assumption, and given arbitrary paths of interest rates and wages for 100 years, backward recursion can be used to compute the corresponding excess demands for bonds and for labor in each year. Our computed equilibrium transition paths of interest rates and wages, consisting of 200 prices in total, are the paths which minimize the sum of squared excess demands. The computation of equilibrium price paths therefore involves minimizing a complicated function of 200 variables.
	
	To achieve this minimization we developed a recursive scheme involving cubic splines. We initially specified the price paths to be cubic splines with knots at years $1,5,10,20,50,100$, and equal to the stationary equilibrium prices in year $100$. This reduces the minimization from $200$ variables to $10$ variables. We found this to be computationally feasible using linear transition paths as starting values. Then we repeated the procedure with knots at one additional year, using the previously computed cubic spline to provide starting values. We continued sequentially adding knots in this way until there were knots for every year up to year $25$, then every five years up to year $50$, then every $10$ years up to year $100$. The resulting computation of equilibrium price paths achieved an absolute excess demand for bonds no greater than $0.0007$ in each year, and an absolute excess demand for labor no greater than $0.0024$ in each year.
	
	\section{Concluding remarks}\label{sec:conclusion}
	
	The analysis presented in this article has obvious implications for tax policy, though the usual caveat about models versus reality applies. Our findings support the complete replacement of labor income taxation with consumption taxation. They leave the door open for capital income to form a significant part of the tax base, but the taxation of capital income should be applied with full offset provisions so that the Domar-Musgrave effect mitigates the disincentive to invest. Full offset provisions may be implemented by direct transfers to entities reporting a loss, or indirectly through loss carry-forward provisions. The latter implementation is already present to varying degrees in different tax jurisdictions.
	
	A skeptical reader may object that the progressivity of consumption taxation relative to labor income taxation is exaggerated in our model due to the assumption that all labor earns the same wage and is taxed at the same rate. This is true, at least insofar as we are concerned with economies in which labor income is taxed at an increasing marginal rate, such as the United States. We counter that fully one third of wealth in the United States is held by the wealthiest $1\%$ of households. \cite{CagettiDeNardi2006} report that business owners and the self employed comprise $81\%$ of these households. No labor income tax, no matter how progressively implemented, can effectively draw revenue from the top third of household wealth. It can at best be designed to target the middle third of wealth rather than the bottom third. Consumption and capital income taxation are effective tools for drawing revenue from the top third of wealth.
	
	The taxation of wealth has been topical in recent years. Academic research exploring the potential benefits of direct wealth taxation includes \cite{GuvenenKambourovKuruscuOcampoChen2023} and \cite{BoarMidrigan2023}. These articles focus on the direct taxation of wealth as an alternative to the taxation of capital income, reaching opposite conclusions about which should be preferred. In our model, it is possible to view the consumption tax as an implicit tax on wealth. The reason is that the optimal consumption rule for agents (Proposition \ref{prop:optrule}) has them consuming fraction $(1-\beta)/(1+\tau_\mathrm{C})$ of their total wealth each period. With $\beta=0.96$ and $\tau_\mathrm{C}=0.31$, the total consumption tax paid by an agent each period is equal to roughly $1\%$ of their total wealth. Total wealth includes not only financial wealth but also the present value of all future labor income, so the amount of consumption tax paid is significant even for agents with no financial wealth. A tax levied only on financial wealth -- which is what is generally meant by a tax on wealth -- ought to be both more progressive and more distortionary than a revenue-equivalent tax levied on total wealth. Future research may explore the possibility of directly taxing financial wealth within the context of our model.
	
	A rate of consumption taxation in excess of 30\%, as we have proposed here, may seem incredible to readers habituated to taxation policy in the United States. The United States is the only major advanced economy without a centrally administered consumption tax. Most component states levy a small consumption tax, with California applying the highest rate of $7.25\%$. Beyond the United States, many nations levy a consumption tax of $20\%$ or more, usually administered as a value-added tax. Such nations include a large majority of European Union members, as well as Brazil, Norway, Russia, Turkey and the United Kingdom, among others. It is our opinion that welfare in the United States would be greatly improved by introducing a substantial nationwide value-added tax on consumption while simultaneously phasing out the taxation of labor income.
	
	\appendix

	\section{Proofs}\label{sec:proofwealth}
	
	Here we provide proofs of the numbered mathematical statements in Section \ref{sec:model}.
	
	\begin{proof}[Proof of Lemma \ref{lem:optimal}]
		When $V_n(s)=a_ns$ the logarithm of the maximand in \eqref{eq:vna} simplifies to
		\begin{align}\label{eq:logmax}
			(1-\beta)\log c+\beta\log(s-(1+\tau_\mathrm{C})c)+\beta\log\nu_\gamma^{-1}(g_n(\theta;a)).
		\end{align}
		We thus maximize the maximand by setting $\theta=\theta_n(a)$ and, by elementary calculus, $c=[(1-\beta)/(1+\tau_\mathrm{C})]s$. Substituting these values of $\theta$ and $c$ into \eqref{eq:logmax} gives
		\begin{align*}
			(1-\beta)\log\left(\frac{1-\beta}{1+\tau_\mathrm{C}}s\right)+\beta\log(\beta s)+\beta\log\kappa_n(a).
		\end{align*}
		Taking the exponential and simplifying yields the maximum value claimed.
	\end{proof}
	
	\begin{proof}[Proof of Proposition \ref{prop:optrule}]
		We first establish the existence of a unique solution $a=a^\ast$ to \eqref{eq:nlsystem}. Define $x_n=\log a_n$, $x=(x_1,\dots,x_N)$, and the map $T:\R^N\to \R^N$ by
		\begin{equation}\label{eq:Tdef}
			T_nx=(1-\beta)\log\frac{1-\beta}{1+\tau_\mathrm{C}}+\beta\log\beta+\beta \log\kappa_n(\exp(x)),
		\end{equation}
		where $\exp$ applies to vectors entry-wise. Since $\kappa_n(a)$ is monotone in $a$, $T$ is also monotone. Furthermore, due to the fact that $\kappa_n(a)$ is positive homogeneous of degree one in $a$, for any constant $k\ge 0$ we have $T_n(x+k1_N)=T_nx+\beta k$. Blackwell's sufficient condition \citep[see e.g.][Thm.~3.3]{StokeyLucas1989} thus implies that $T$ is a contraction with modulus $\beta<1$, and so the contraction mapping theorem implies the existence of a unique fixed point $x^\ast$ of $T$. The existence of a unique solution $a=a^\ast$ to \eqref{eq:nlsystem} follows. Moreover, Lemma \ref{lem:optimal} implies that the value function $V^\ast_n(s)=a_n^\ast s$ solves the Bellman equation \eqref{eq:vna}.
		
		It remains to verify the decision rules. The consumption rule \eqref{eq:crule} is immediate from Lemma \ref{lem:optimal}. The capital rule \eqref{eq:krule} follows from \eqref{eq:frack} and \eqref{eq:crule}. The labor rule \eqref{eq:lrule} follows from \eqref{eq:krule} and the discussion surrounding \eqref{eq:elln} and \eqref{eq:maxW}. The bond rule \eqref{eq:brule} follows from \eqref{eq:fracb}, \eqref{eq:crule}, and \eqref{eq:humanwealth}.
	\end{proof}
	
	Lemma \ref{lem:optimal} and Proposition \ref{prop:optrule} together establish that the solution $V^\ast_n(s)=a_n^\ast s$ to the Bellman equation is unique in the class of candidate value functions of the form $V_n(s)=a_ns$. The solution is in fact unique in the broader class
	\begin{equation*}
		\cV\coloneqq\left\{V\in\mathcal B((0,\infty)\times \cN):0<\inf_{s,n}\frac{V_n(s)}{s}\leq\sup_{s,n}\frac{V_n(s)}{s}<\infty\right\},
	\end{equation*}
	where $\mathcal B((0,\infty)\times \cN)$ is the collection of real Borel functions on $(0,\infty)\times \cN$.
	\begin{prop}\label{pro:uniqueV}
		The value function $V^\ast_n(s)=a_n^\ast s$ is the unique solution to the Bellman equation \eqref{eq:vna} in $\cV$.
	\end{prop}
	\begin{proof}[Proof of Proposition \ref{pro:uniqueV}]
		In view of the definition of $\cV$, for any $V\in\cV$ we may choose $N\times1$ vectors $\ubar{a},\bar{a}$ with positive entries such that $\ubar{a}_ns\leq V_n(s)\leq\bar{a}_ns$ for all $s\in(0,\infty)$ and $n\in\cN$. Define the map $\tilde{T}:\cV\to\mathcal B((0,\infty)\times\cN)$ by setting $(\tilde{T}V)_n(s)$ equal to the right-hand side of the Bellman equation \eqref{eq:vna}. Using the monotonicity of $\tilde{T}$ and applying Lemma \ref{lem:optimal}, we obtain
		\begin{equation}\label{eq:TVbounds}
			\left(\frac{1-\beta}{1+\tau_\mathrm{C}}\right)^{1-\beta}(\beta\kappa_n(\ubar{a}))^\beta s \le (\tilde{T}V)_n(s)\le \left(\frac{1-\beta}{1+\tau_\mathrm{C}}\right)^{1-\beta}(\beta\kappa_n(\bar{a}))^\beta s.
		\end{equation}
		Dividing by $s>0$ and taking the infimum and supremum over $s$ and $n$, we see that $\tilde{T}V\in \cV$. We may thus restrict the codomain of $\tilde{T}$ to $\cV$ and write $\tilde{T}:\cV\to\cV$.
		
		A candidate value function in $\cV$ solves the Bellman equation \eqref{eq:vna} if and only if it is a fixed point of $\tilde{T}$. Proposition \ref{prop:optrule} establishes that one such solution is given by $V^\ast_n(s)=a^\ast_ns$. It remains to show uniqueness. Let $V\in \cV$ be any fixed point of $\tilde{T}$, and rewrite \eqref{eq:TVbounds} as
		\begin{equation}\label{eq:TVbounds2}
			\exp(T_n(\log\ubar{a}))s\leq V_n(s)\leq\exp(T_n(\log\bar{a}))s,
		\end{equation}
		where $T_n$ is as defined in \eqref{eq:Tdef}, and $\log$ applies to vectors entry-wise. Lemma \ref{lem:optimal} reveals that, for $k\in\N$, applying $\tilde{T}^k$ to $\exp(T_n(\log a))s$ (viewed as a function of $s$ and $n$) gives $\exp(T^{k+1}_n(\log a))s$. Therefore, applying $\tilde{T}^k$ to \eqref{eq:TVbounds2}, we obtain
		\begin{equation*}
			\exp((T^{k+1})_n(\log \ubar{a}))s\le V_n(s)\le \exp((T^{k+1})_n(\log \bar{a}))s.
		\end{equation*}
		Letting $k\to\infty$ and noting that $T$ is a contraction with fixed point $x^\ast=\log a^\ast$ (as established in the proof of Proposition \ref{prop:optrule}), we have $T^{k+1}(\log a)\to\log a^*$ for any $a$, and so $V_n(s)=a_n^*s=V_n^*(s)$.
	\end{proof}
	\begin{proof}[Proof of Proposition \ref{prop:stationary}]
		To account for the fact that a Markov multiplicative process with reset as defined in \citet{BeareToda2022} resets to one, as opposed to $S_t$ resetting to $h$ in our model, we define the scaled wealth process $\tilde{S}_t\coloneqq h^{-1}S_t$, which resets to one. The sequence of pairs $(\tilde{S}_t,J_t)_{t\in\mathbb Z_+}$ is then a Markov multiplicative process with reset, and has a unique stationary distribution by Proposition 3 in \citet{BeareToda2022}. The pair $(h^{-1}S,J)$ is a random draw from this stationary distribution. For all complex $z$ with real part belonging to $\mathcal I_-$, Lemma 2 in \citet{BeareToda2022} implies that $\mathrm{I}-\A(z)$ is invertible, that
		\begin{align*}
			\E((h^{-1}S)^z\mid J=n)=(1-\upsilon)p_n^{-1}\varpi^\top(\mathrm{I}-\A(z))^{-1}e^{(n)}
		\end{align*}
		for each $n\in\mathcal N$ with $p_n>0$, and that
		\begin{align*}
			\E((h^{-1}S)^z)=(1-\upsilon)\varpi^\top(\mathrm{I}-\A(z))^{-1}1_N.		
		\end{align*}
		Equations \eqref{eq:conditionalwealthMellin} and \eqref{eq:wealthMellin} follow immediately. If the equation $\rho(\A(z))=1$ admits a unique positive solution $z=\alpha$, then Theorem 1 in \citet{BeareToda2022} implies that
		\begin{align}\label{eq:hPareto}
			\lim_{s\to\infty}\frac{\log\mathrm{P}(h^{-1}S>s)}{\log s}&=-\alpha.
		\end{align}
		Noting that
		\begin{align*}
			\lim_{s\to\infty}\frac{\log\mathrm{P}(S>s)}{\log s}&=\lim_{s\to\infty}\frac{\log\mathrm{P}(S>hs)}{\log hs}\\&=\left(\lim_{s\to\infty}\frac{\log\mathrm{P}(h^{-1}S>s)}{\log s}\right)\left(\lim_{s\to\infty}\frac{\log hs-\log h}{\log hs}\right),
		\end{align*}
		we deduce that \eqref{eq:hPareto} implies \eqref{eq:Pareto}.
	\end{proof}
	\begin{proof}[Proof of Proposition \ref{prop:equilibrium}]
		If $\rho(\A(1))<1$ then Proposition \ref{prop:stationary} establishes that, for each $n\in\mathcal N$ with $p_n>0$,
		\begin{align*}
			\E(S\mid J=n)&=(1-\upsilon)p_n^{-1}h\varpi^\top(\mathrm{I}-\A(1))^{-1}e^{(n)}.
		\end{align*}
		We thus deduce from \eqref{eq:krule}, \eqref{eq:lrule} and \eqref{eq:brule} in Proposition \ref{prop:optrule} that
		\begin{align*}
			\E(B_J^*(S)\mid J=n)&=-\frac{h}{R}+\frac{\beta}{\upsilon}(1-\theta_n^*)(1-\upsilon)p_n^{-1}h\varpi^\top(\mathrm{I}-\A(1))^{-1}e^{(n)},\\
			\E(L_J^*(S)\mid J=n)&=\frac{\beta}{\upsilon}\theta_n^*\ell_n(\omega)(1-\upsilon)p_n^{-1}h\varpi^\top(\mathrm{I}-\A(1))^{-1}e^{(n)}.
		\end{align*}
		Multiplying both equations by $p_n$ and summing over $n$ yields \eqref{eq:aggB} and \eqref{eq:aggL}.
		
		Proposition 1 in \citet{BeareToda2022} implies that $\rho(\A(z))$ is a convex function of real $z$ satisfying $\rho(\A(0))=\upsilon<1$, so if $\rho(A(1))\geq1$ then the equation $\rho(\A(z))=1$ must admit a unique positive solution $z=\alpha$ with $\alpha\leq 1$. It then follows from Proposition \ref{prop:stationary} that the right tail of the stationary distribution of wealth is Pareto with decay rate $\alpha\leq1$, implying that $\E(S)=\infty$.
	\end{proof}
	\begin{proof}[Proof of Proposition \ref{prop:welfare}]
		For each real $z$ we have
		\begin{align}
			\E(V_J^\ast(h)^z)&=\sum_{n=1}^N\varpi_n(a_n^\ast h)^z =h^z\varpi^\top(a^\ast)^z.\label{eq:welfareproof}
		\end{align}
		If $\gamma\neq1$ then we obtain \eqref{eq:Wnew} from \eqref{eq:welfareproof} by setting $z=1-\gamma$ and raising to the power of $1/(1-\gamma)$. If $\gamma=1$ then we observe that \eqref{eq:welfareproof} provides a formula for the moment generating function of $\log V^\ast_J(h)$. The derivative of this function is
		\begin{align*}
			\frac{\mathrm{d}}{\mathrm{d}z}\E(V^\ast_J(h)^z)&=(\log h)h^z\varpi^\top(a^\ast)^z+h^z\varpi^\top((a^\ast)^z\odot\log a^\ast).
		\end{align*}
		Therefore,
		\begin{align*}
			\E(\log V_J^\ast(h))&=\left.\frac{\mathrm{d}}{\mathrm{d}z}\E(V^\ast_J(h)^z)\right|_{z=0}=\log h+\varpi^\top\log a^\ast.
		\end{align*}
		Taking the exponential yields \eqref{eq:Wnew} for the case $\gamma=1$.
	\end{proof}
	\begin{proof}[Proof of Proposition \ref{prop:tax}]
		The formula for $\E(T_\mathrm{C}(S))$ follows immediately from the linearity of $T_\mathrm{C}(s)$ in $s$ and the formula for $\E(S)$ obtained by setting $z=1$ in Proposition \ref{prop:stationary}. To obtain the formula for $\upsilon\E(T_\mathrm{K}(S,J))$ we first observe that, for each $n\in\cN$ with $p_n>0$,
		\begin{multline*}
			\upsilon\E(T_\mathrm{K}(S,J)\mid J=n)=\frac{\tau_\mathrm{K}}{1-\tau_\mathrm{K}}\beta\theta^\ast_n\E(S\mid J=n)\sum_{n'=1}^N\pi_{nn'}r_{n'}(\omega)\\=\frac{\tau_\mathrm{K}}{1-\tau_\mathrm{K}}\beta (1-\upsilon)hp_n^{-1}\sum_{n'=1}^N(\varpi^\top(\mathrm{I}-\A(z))^{-1}e^{(n)})\theta^\ast_n\pi_{nn'}r_{n'}(\omega),
		\end{multline*}
		where the second equality follows from the formula for $\E(S\mid J=n)$ obtained by setting $z=1$ in Proposition \ref{prop:stationary}. It now follows from the law of iterated expectations that
		\begin{align*}
			\upsilon\E(T_\mathrm{K}(S,J))&=\frac{\tau_\mathrm{K}}{1-\tau_\mathrm{K}}\beta (1-\upsilon)h\sum_{n=1}^N\sum_{n'=1}^N(\varpi^\top(\mathrm{I}-\A(z))^{-1}e^{(n)})\theta^\ast_n\pi_{nn'}r_{n'}(\omega)\\
			&=\frac{\tau_\mathrm{K}}{1-\tau_\mathrm{K}}\beta(1-\upsilon)h(\varpi^\top(\mathrm{I}-\A(z))^{-1}\odot\theta^{\ast\top})\Pi r\\&=\frac{\tau_\mathrm{K}}{1-\tau_\mathrm{K}}\beta(1-\upsilon)h\varpi^\top(\mathrm{I}-\A(z))^{-1}(\Pi r\odot\theta^\ast),
		\end{align*}
		as claimed.
	\end{proof}

	
\end{document}